\makeatletter 

\documentclass[aps,twocolumn,pra,tightenlines,floatfix,showpacs,superscriptaddress]{revtex4-1}
\usepackage{graphicx}
\usepackage{epstopdf}
\usepackage[english]{babel}
\usepackage{amsmath}
\usepackage{ntheorem}
\theoremseparator{:}
\theoremheaderfont {\it}
\newtheorem*{lemma}{\indent Lemma}
\newtheorem*{proposition}{\indent Proposition}
\newenvironment{proof}{{\noindent\it Proof:}}{\hfill $\blacksquare$\par}
\usepackage{amssymb}
\usepackage{times}
\usepackage{appendix}
\usepackage{bm}
\usepackage{mathrsfs}   
\usepackage{float}
\usepackage{color}
\usepackage{longtable}
\usepackage{url}
\usepackage[usenames,dvipsnames]{xcolor}
\usepackage{subfigure}
\usepackage[percent]{overpic}
\renewcommand{\@thesubfigure}{\hskip\subfiglabelskip}
\usepackage{indentfirst} 
\graphicspath{{./Figures}}

\newcommand{\Rmnum}[1]{\expandafter\@slowromancap\romannumeral #1@}
\newcommand{\Tr}{\text{Tr}}
\usepackage{array}
\usepackage{booktabs}

\makeatother 

\usepackage[colorlinks=true,linkcolor=Blue,urlcolor=Blue,citecolor=Blue,anchorcolor=Red]{hyperref}
\usepackage{cleveref}
\crefname{equation}{Eq.~}{Eq.~}
\crefname{figure}{Fig.~}{Fig.~}
\begin{document}

\title{Spin squeezing and concurrence under Lee-Yang dephasing channels}

\author{Yuguo Su}

\affiliation{Zhejiang Province Key Laboratory of Quantum Technology and Device, Department of Physics, Zhejiang University, Hangzhou, Zhejiang 310027, China}

\author{Hongbin Liang}

\affiliation{Zhejiang Province Key Laboratory of Quantum Technology and Device, Department of Physics, Zhejiang University, Hangzhou, Zhejiang 310027, China}

\author{Xiaoguang Wang}

\email{xgwang1208@zju.edu.cn}

\affiliation{Zhejiang Province Key Laboratory of Quantum Technology and Device, Department of Physics, Zhejiang University, Hangzhou, Zhejiang 310027, China}
\begin{abstract}
The Lee-Yang zeros are one-to-one mapped to zeros in the coherence of a probe spin coupled to a many-body system.
Here, we propose two different types of Lee-Yang dephasing channels, where the partition functions vanish at Lee-Yang zeros, and study the spin squeezing under them.
Under the first type of the channels where probes are coupled to their own bath, we find that
the performance of spin squeezing is improved and its maximum only depends on the initial state.
Moreover, the centers of all the concurrence vanishing domains are corresponding to the Lee-Yang zeros.
Under the second type of the channel where probes are coupled to one bath together, the performance of spin squeezing is not improved, however, the concurrence shares almost the same properties under both channels.
These results provide new experimental possibilities in many-body physics and extend a new perspective of the relationship between the entanglement and spin squeezing in probes-bath systems.
\end{abstract}
\maketitle

\section{Introduction}
In 1952, when Lee and Yang studied the statistical theory of equations of state and phase transitions,  they proved that the partition functions of thermal systems exist zero roots, named Lee-Yang zeros, on the complex plane of fugacity or a magnetic field \cite{Yang-PR87.404--409-1952,Lee-PR87.410--419-1952}.
They provided an insight into the thermodynamic properties of the Ising ferromagnet at arbitrary temperature in an arbitrary nonzero external fugacity or magnetic field.
As a cornerstone of statistical mechanics, they revealed the fact that under very general conditions, all the Lee-Yang zeros of a general Ising ferromagnet lie on the unit circle in the complex plane, which is the so-called unit-circle theorem \cite{Lee-PR87.410--419-1952}.
According to the celebrated theorem, under the thermodynamic limit condition, the Lee-Yang zeros form a continuum ring in the complex plane \cite{Yang-PR87.404--409-1952}.
Above the critical temperature, the continuum ring fractures and has a gap around the positive real axis which is a forbidden zone of the roots of the partition function \cite{Yang-PR87.404--409-1952}, in other words, the free energy is analytic and there is no phase transition.
Moreover, as the temperature decreases, the gap becomes narrowing and the two edge points approach the real axis at the critical temperature \cite{Yang-PR87.404--409-1952}.
Fisher \cite{Fisher-PRL40.1610--1613-1978} presented the concept, Yang-Lee edge singularities, which mean that the two edge points of the broken ring are singularity points \cite{Kortman-PRL27.1439--1442-1971}.
Liu \cite{Wei-PRL109.185701-2012} showed that the relationship between Lee-Yang zeros and the zeros in the coherence \cite{Schlosshauer-RMP76.1267--1305-2005,Liu-NJP9.226-2007} of a probe spin coupled to the many-body system is bijective.
Furthermore, with the temperature overtopping the critical point, sudden death and birth of the coherence occurring at critical times are corresponding to the Yang-Lee singularities in the thermodynamic limit.

 Because of the intrinsic difficulty that Lee-Yang zeros would occur only at complex values of physical parameters, which are generally regarded as unphysical, physicists deem it hard to observe them. However, the Lee-Yang zeros have been observed in experiments \cite{Peng2015}, recently.
 The applications of the Lee-Yang theorem are manifold in general ferromagnetic Ising models of arbitrarily high spin \cite{Asano-PTP40.1328-1336-1968,Suzuki-PTP40.1246-1256-1968,Griffiths-JMP10.1559-1565-1969}, antiferromagnetic Ising models \cite{Kim2004}, forecasting the large-deviation statistics of the activity \cite{Brandner2017} as well as other striking types of interactions \cite{Suzuki-JMP9.2064-2068-1968,Suzuki-JMP12.235-246-1971,Kurtze-JSP19.205--218-1978}.

In the past decades, spin squeezing has attracted lots of attention \cite{Kitagawa1993,Wineland1994,Sorensen2001,Toth2007,Toth2009,Sorensen1999,Wang2003,Ma2011}.
It is well known that there are close relations between entanglement and spin squeezing, and a lot of effort has been devoted to unveiling it \cite{Ulam-Orgikh2001,Jin2007,Jafarpour,Messikh2003,Wang2004,Yin2010}.
Because spin squeezing is relatively easy to be generated and measured experimentally \cite{Genes2003,Fernholz2008,Takano2009}, spin-squeezing parameters are promising candidates as measures of many-body correlations.
Improving the precision of measurements is another important application of spin squeezing.
For example, spin squeezing plays an important role in Ramsey spectroscopy \cite{Wineland1992,Wineland1994,Cronin2009,Bollinger1996,Doering2010}, as well as in making high-precision atomic clocks \cite{Sorensen1999,Andre2004,Meiser2008} and gravitational-wave interferometers \cite{Walls,Goda}, etc.

\begin{figure}[htbp]
\begin{minipage}{1\linewidth}
\begin{overpic}[width=\linewidth]{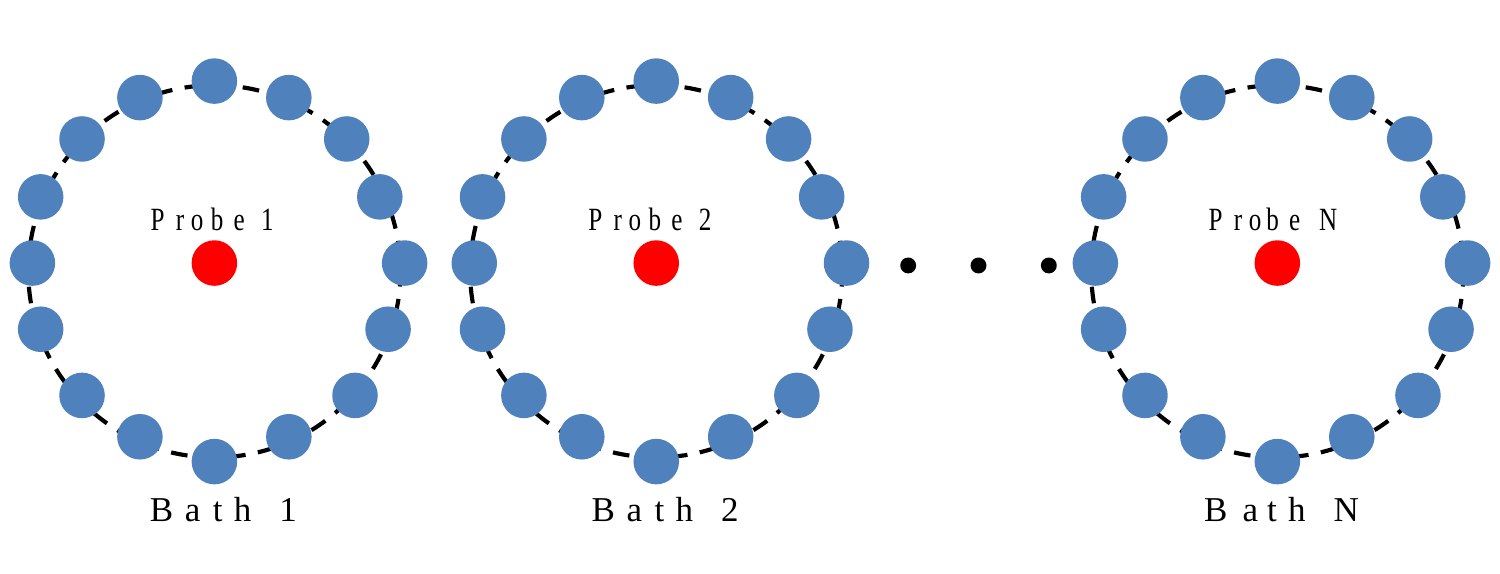}\label{fig1-a}
\put(0,35){$(a)$}
\end{overpic}
\end{minipage}\\
\begin{minipage}{0.6\linewidth}
\begin{overpic}[width=\linewidth]{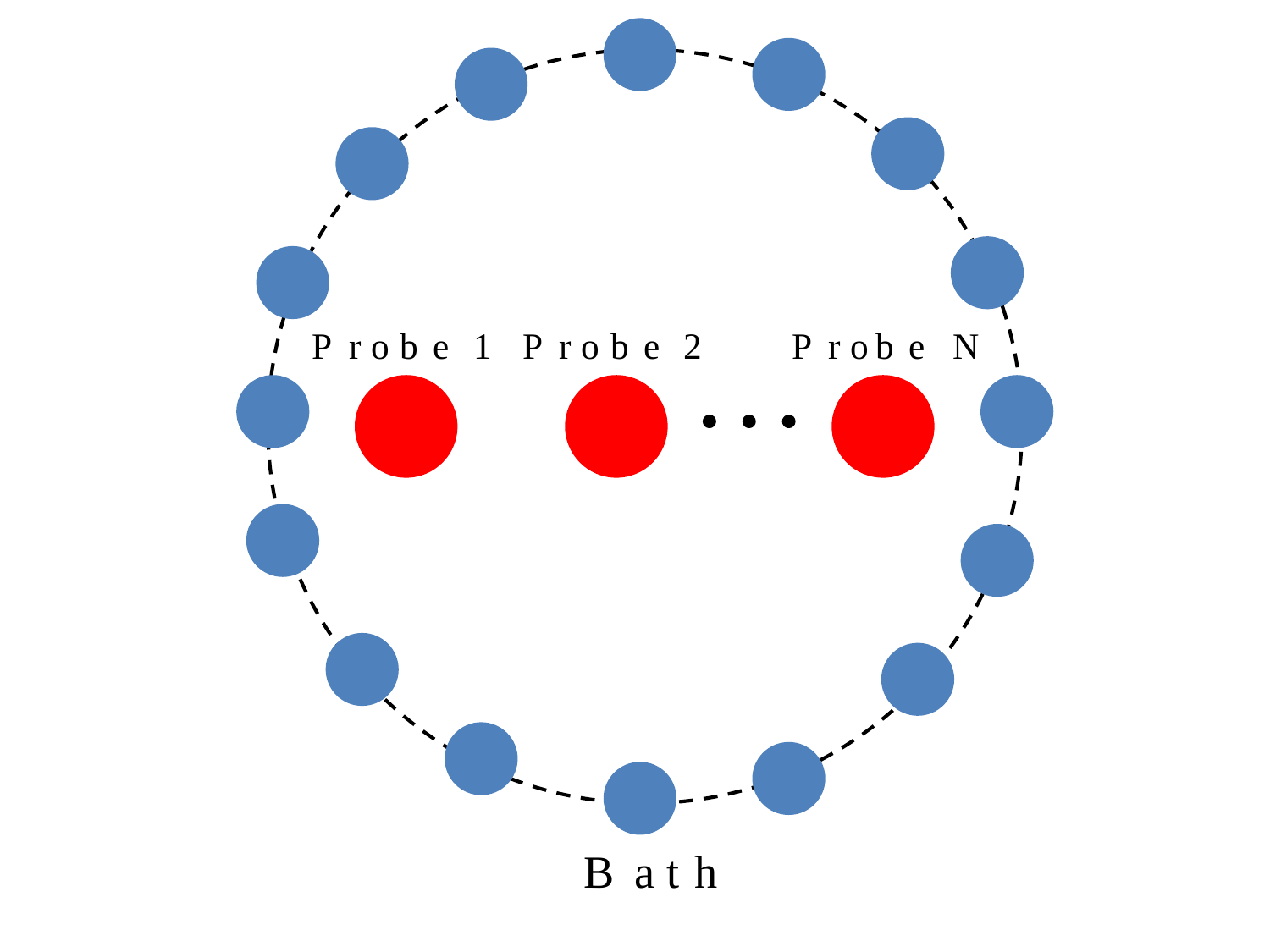}\label{fig1-a}
\put(0,65){$(b)$}
\end{overpic}
\end{minipage}
\protect\protect\caption{
Two different dephasing channels under a probe(s)-bath system.
(a) $N$ dephasing channels where probes are only respectively coupled to their own bath.
(b) The dephasing channel where $N$ probes are coupled to one bath together.
The red circles mark probes and the blue circles mark bath spins.
}
 \label{fig1}
\end{figure}


There are a lot of merits of spin squeezing in measurements, however, we have not seen the work that combines spin squeezing and Lee-Yang zeros.
In this work, we propose two different types of Lee-Yang dephasing channels, where the partition functions vanish at Lee-Yang zeros, and study the spin squeezing under them.
Under the first type of the channels where probes are coupled to their own bath [Fig.~1(a)], we calculate the coherence, the rescaled concurrence and the spin-squeezing parameter under the channels.
We find that the performance of spin squeezing is improved and its maximum only depends on the initial state.
Moreover, the centers of all the concurrence vanishing domains are corresponding to the Lee-Yang zeros.
Under the second type of the channel where probes coupled to one bath together [Fig.~1(b)], the performance of spin squeezing is not improved. However, the concurrence shares almost the same properties under both channels.
Furthermore, under the latter channel, the coherence could reach the stable value $\left(1-\cos^{N-2}\theta\right)/4$ at times in correspondence to the Lee-Yang zeros, where $\theta$ is the twist angle and $N$ is the probe spin number.
These results provide new experimental possibilities in many-body physics and extend a new perspective of the relationship between the entanglement and spin squeezing in probes-bath systems.


This paper is organized as follows.
In Sec.~\ref{Sec. II}, we introduce the Lee-Yang zeros and the dephasing channels.
The initial state of the probe-bath systems and definitions of the spin-squeezing and the concurrence are given in Sec.~\ref{Sec. III}.
In Sec.~\ref{Sec. IV}, we introduce the two different types of  Lee-Yang dephasing channels and prove that the probe(s)-bath system whose ferromagnetic Ising bath is under zero field is a Lee-Yang dephasing channel.
The coherence, the rescaled concurrence and the spin-squeezing parameter of multiprobes coupled to their own bath are given in Sec.~\ref{Sec. V}, and we discuss the influence of the bath spins number $N_b$, the inverse temperature $\beta$, the probes number $N$ and the twisted angle $\theta$ on the them.
In Sec.~\ref{Sec. VI}, we calculate and analyse the coherence, the rescaled concurrence and the spin-squeezing parameter of multiprobes coupled to one bath together.
Finally, a summary is given in Sec.~\ref{Sec. VII}.
\section{Lee-Yang Zeros and Dephasing Channels}\label{Sec. II}
We focus on a general Ising model with ferromagnetic interactions under a magnetic field $h$.
The Hamiltonian \cite{Wei-PRL109.185701-2012} is
\begin{equation}
H_0\left(h\right)=-\sum_{i,j}\lambda_{ij}s_{i}s_{j}-h\sum_{i}s_{i},
\end{equation}
where the spins $s_{i}$ take values $\pm1$ and the interactions $\lambda_{ij}\geq0$.
If the external magnetic field does not exist, all spins in state $-1$ or $+1$ is the ground state, which means it is completely in order.
Expanded by an $N$th order polynomial of $z\equiv \exp\left(-2\beta h\right)$, the partition function of $N_b$ spins at temperature $T$ can be rewritten as
\begin{equation}
Z\left(\beta,h\right)=\Tr\left[e^{-\beta H}\right]=e^{\beta N_bh}\sum_{n=0}^{N_b}f_{n}z^{n},
\end{equation}
where $\beta=1/T$ is the inverse temperature (Boltzmann and Planck constants taken unity) and $f_{n}$ is the partition function with zero magnetic field under the condition that $n$ spins are in the state $-1$ (equivalently, we can also take the $n$ spins in the state $+1$, since we know $f_{n}=f_{N_b-n}$).
In principle, all physical properties of an exotic system can be calculated from the partition function.
Obviously, according to the notable unit-circle theorem, the $N_b$ zeros of the partition function, which are located on the unit circle in the complex plane of $z$ \cite{Lee-PR87.410--419-1952}, can be rewritten as $z_{n}\equiv \exp\left(\text{i}\phi_{n}\right)$ with $n=1,2,\ldots,N_b$.
Supplied the Lee-Yang zeros as a priori knowledge, the partition function can be expressed as
\begin{equation}
Z\left(\beta,h\right)=f_{0}e^{\beta N_bh}\prod_{n=1}^{N_b}\left(z-z_{n}\right),
\end{equation}
by employing the polynomial factorization.
Observing the equation, we can easily know that we can't gain the $z_{n}$ or $\phi_{n}$ straightway.
However, we could prove that the partition function is a real and even function in a ferromagnetic Ising bath under zero field (see Appendix \ref{Appendix_A}).

Here, we introduce the dephasing channels.
We write the total Hamiltonian of the system as follows:
\begin{align}
H=H_0\otimes I_{2\times2}+\eta H_1\otimes \sigma_z,
\end{align}
where $H_0$ is the Hamiltonian of the environment, $\eta H_1\otimes \sigma_z$ is the Hamiltonian of the spin-environment interaction, $\eta$ is a coupling constant, $H_1=-\sum_{j}s_{j}$ acts as the random field for the probe spin, $\sigma_{z}=\left(|0\rangle\langle0|-|1\rangle\langle1|\right)$ is the Pauli matrix and $I_{2\times2}$ is the identity operator.

Therefore, we have the unitary matrix
\begin{align}
U\equiv \text{e}^{-\text{i}Ht}=\text{e}^{-\text{i}H_0t}\text{e}^{-\text{i}\eta \sigma_zH_1t}.\label{eq:U}
\end{align}
Employing the unitary transformation, we can gain the time evolution
density matrix of the system
\begin{equation}
\rho\left(t\right) =U\left(\rho_{B}\left(0\right)\otimes\rho_{P}\left(0\right)\right)U^{\dagger},\label{rho_t}
\end{equation}
where
\begin{align}
  \rho_{P}\left(0\right)=\rho_{00}|0\rangle\langle0|+\rho_{01}|0\rangle\langle1|+\rho_{10}|1\rangle\langle0|+\rho_{11}|1\rangle\langle1|
\end{align}
and $\rho_{B}\left(0\right)$ are the initial density matrices of the spin and the environment, respectively.

Hence, we can calculate the elements of the spin density matrix as follows:
\begin{align}
T_{ij}  =  \Tr_B\big[\langle i|U\left(\rho_{B}\left(0\right)\otimes\rho_{P}\left(0\right)\right)U^{\dagger}|j\rangle\big] , \label{eq:T11_1}
  \end{align}
where $\sigma_z|j\rangle=\left(\text{i}\right)^{2j}|j\rangle$ $\left( i,j=0,1\right)$, and
$\exp\left(\sigma_{z}\right)=\sum_{n=0}^{\infty}\sigma_{z}^{n}/n!$.
Then we find that
\begin{align}
\text{e}^{\pm\text{i}\eta\sigma_z H_1t}|j\rangle &= \text{e}^{\pm\left(\text{i}\right)^{2j+1}\eta H_1t}|j\rangle,\label{eq:simplify_1}
\end{align}
and we can rewrite Eq.~\eqref{eq:T11_1}, for example,
\begin{equation}
T_{00}  = \rho_{00}\Tr_B[\text{e}^{-\text{i} H_0t}\frac{e^{-\beta H_0}}{Z\left(\beta,h\right)}\text{e}^{\text{i} H_0t}]=\rho_{00}.
\end{equation}
Following this method, we gain the evolution of the reduced initial density matrix
\begin{align}
  \rho \left(t\right) =\begin{pmatrix}
    \rho_{00} & \mu\rho_{01}\\
    \nu\rho_{10} & \rho_{11}\\
  \end{pmatrix},
\end{align}
where the coefficients
\begin{align}
\mu&=\Tr_B\left[\frac{e^{-\beta H_0}}{Z\left(\beta,h\right)}\text{e}^{-2\text{i}\eta H_1t}\right],\\
\nu&=\Tr_B\left[\frac{e^{-\beta H_0}}{Z\left(\beta,h\right)}\text{e}^{2\text{i}\eta H_1t}\right].
\end{align}

Comparing with the definition of the dephasing channel
\begin{align}
  \mathcal{E}\begin{pmatrix}
    \rho_{00} & \rho_{01}\\
    \rho_{10} & \rho_{11}\\
  \end{pmatrix}=\begin{pmatrix}
    \rho_{00} & \left(1-p\right)\rho_{01}\\
    \left(1-p\right)\rho_{10} & \rho_{11}\\
  \end{pmatrix},
\end{align}
one can easily find that $\rho\left(t\right)$ is a quantum dephasing channel when $p=1-\mu=1-\nu$, and the three Kraus operators are given by
\begin{align}
  M_0\!=\!\sqrt{\mu}I, M_1\!=\!\sqrt{1\!-\!\mu}|0\rangle\langle0|,M_2\!=\!\sqrt{1\!-\!\mu}|1\rangle\langle1|,
\end{align}
where $I$ is the identity operator.
In the Heisenberg picture, a quantum channel with the Kraus operators is defined via the map
\begin{align}
  \mathcal{E}\left(\rho\right)=\sum_{i}{M_i\rho M_i^{\dagger}},\label{QC}
\end{align}
and an expectation value of the operator $O$  can be gained by
\begin{align}
  \left\langle O\right\rangle=\Tr\left[O\mathcal{E}\left(\rho\right)\right]=\Tr\left[\mathcal{E}^{\dagger}\left(O\right)\rho\right].\label{O}
\end{align}
\section{Initial State,  and Definitions of Spin-squeezing  and Concurrence}\label{Sec. III}
With so many benefits of spin squeezing, we are driven to apply spin squeezing to dephasing channels.
We consider an ensemble of $N$ spin-$1/2$ probes with ground state $|1\rangle$ and excited state $|0\rangle$.
To take advantage of exchange symmetry, we choose the one-axis twisted state
\begin{align}
  |\psi\left(0\right)\rangle=\text{e}^{-\text{i}\theta J_x^2/2}|1\rangle^{\otimes N}\equiv\text{e}^{-\text{i}\theta J_x^2/2}|\textbf{1}\rangle,
\end{align}
where $N$ is the number of the total qubits, the state is prepared by the one-axis twisting Hamiltonian $H=\chi J_x^2$, with the coupling constant $\chi$ and $\theta=2\chi t$ being the twist angle.
We set the mean spin of the initial state along the $z$ direction, the two-qubit reduced density matrix becomes
\begin{align}
  \rho_{P}\left(0\right)=\begin{pmatrix}
v_{+}&0&0&u^*\\
0&w&y&0\\
0&y&w&0\\
u&0&0&v_{-}\label{rho_ss}\\
\end{pmatrix}
\end{align}
in the basis $\left\{\left|11\right\rangle,\left|00\right\rangle,\left|01\right\rangle,\left|10\right\rangle\right\}$, where
\begin{align}
  v_{\pm}&=\left(1\pm2\left\langle\sigma_{1z}\right\rangle+\left\langle\sigma_{1z}\sigma_{2z}\right\rangle\right)/4,\label{v_pm}\\
   w&=\left(1-\left\langle\sigma_{1z}\sigma_{2z}\right\rangle\right)/4,\\
  y&=\left\langle\sigma_{1+}\sigma_{2-}\right\rangle,\\
  u&=\left\langle\sigma_{1-}\sigma_{2-}\right\rangle.
\end{align}
For the one-axis twisted state, one could have \cite{Yin2012}
\begin{align}
  \left\langle\!\sigma_{z}\!\right\rangle&=-\cos^{N-1}\left(\frac{\theta}{2}\right),\\
  \left\langle\!\sigma_{1z}\!\sigma_{2z}\right\rangle&=\left(1+\cos^{N-2}\theta\right)/2,\\
  \left\langle\!\sigma_{1+}\!\sigma_{2-}\!\right\rangle&=\frac{1}{8} \left(1-\cos ^{N-2}\theta\right),\\
  \left\langle\!\sigma_{1-}\sigma_{2-}\!\right\rangle&=\!-\frac{1}{8} \!\left(\!1\!-\!\cos ^{N\!-\!2}\!\theta\!\right)\!-\!\frac{\text{i}}{2}\! \sin \! \left(\!\frac{\theta }{2}\!\right) \!\cos ^{N\!-\!2}\!\left(\!\frac{\theta }{2}\!\right).\label{sigma}
\end{align}

Then we discuss the spin-squeezing parameter defined by Ref. \cite{Kitagawa1993} and we rewrite it as follows:
\begin{align}
  \xi^2\equiv\frac{4\left(\Delta J_{\perp}\right)^2_{\text{min}}}{N}.
\end{align}
For our initial state, it can be simplify as
\begin{align}
  \xi^2=1+2\left(N-1\right)\left(\left\langle\sigma_{1+}\sigma_{2-}\right\rangle-\left|\left\langle\sigma_{1-}\sigma_{2-}\right\rangle\right|\right).
\end{align}

The entanglement of formation is defined as follows Ref. \cite{Bennett-PRA54.3824--3851-1996}, and the concurrence quantifying the entanglement of a pair of spin-$1/2$  is defined as
\begin{equation}
C\left(\rho\right)  =\text{max}\left\{ 0,\lambda_{1}-\lambda_{2}-\lambda_{3}-\lambda_{4}\right\}  ,\label{eq:C}
\end{equation}
where the $\lambda_{i}$s are the square roots of the eigenvalues, in decreasing order, of the non-Hermitian matrix
\begin{align}
  R\equiv \rho\left(t\right)\left(\sigma_{1y}\otimes\sigma_{2y}\right)\rho^{*}\left(t\right)\left(\sigma_{1y}\otimes\sigma_{2y}\right),\label{R}
\end{align}
and $\rho^{*}\left(t\right)$ denotes the complex conjugate of $\rho\left(t\right)$.

From Ref. \cite{Wang2003}, we know that if $\xi^2\leq1$ for even and odd states, then we gain the relationship between the spin-squeezing parameter and the rescaled concurrence: $\xi^2=1-\left(N-1\right)C$.
To deliberate on the relationship, we take
\begin{align}
  \xi^{\prime 2}\equiv1-C_r
\end{align}
for convenience, where $C_r=\left(N-1\right)C\left(\rho\left(t\right)\right)\label{Cr}$ is the rescaled concurrence.
Later, We will examine the relationship under the Lee-Yang dephasing channels shown in Sec.~\ref{Sec. IV}.
\section{Two Types of  Lee-Yang Dephasing Channels}\label{Sec. IV}
In this section, we introduce two different types of Lee-Yang dephasing channels shown in Fig.~1.
As shown in Fig.~1(a), dephasing channels are independent of each other, where the probe is only coupled to its own bath.
The dephasing channel has $N$ probes and each probe is equally coupled to all the $N_b$ bath spins [Fig.~1(b)].

\subsection{Dephasing channels where probes are coupled to their own bath}
We consider the dephasing channels where probes are only coupled to its own bath, and choose spin-squeezing probes (spin-$1/2$) and the Ising chains as the probes and the baths.

We select the standard one-axis twisted state as the initial state under the dephasing channels.
Therefore, we have the total Hamiltonian and the unitary matrix
\begin{align}
  H_{\Rmnum{1}}&=\sum_{k=1}^{N}{\left(H_{0,k}\otimes I_{2\times2}+\eta H_{1,k}\otimes\sigma_{k,z}\right)},\label{H_III}\\
U_{\Rmnum{1}}&= \text{e}^{-\text{i}H_{\Rmnum{1}}t}=\prod_{k=1}^{N}{U_k}=\prod_{k=1}^{N}{\text{e}^{-\text{i}H_{0,k}t}\text{e}^{-\text{i}\eta \sigma_{k,z}H_{1,k}t}},\label{U_III}
\end{align}
where $H_{0,k}$ and $H_{1,k}$ denote the Hamiltonian of the $k$th environment and the $k$th random field, respectively.
The probe spin lying in the bath is equally coupled to all the $N_b$ bath spins.

From Ref. \cite{Yin2012}, one could know that the evolution of any finite number of particles is governed only by the local Hamiltonian of the particles and their baths.
Hence, from Eqs.~\eqref{rho_t},~\eqref{eq:T11_1}-\eqref{eq:simplify_1},~\eqref{rho_ss} and~\eqref{H_III}-\eqref{U_III}, the evolution of the reduced initial state is
\begin{align}
  \rho_{\Rmnum{1}}\left(t\right)=\begin{pmatrix}
v_{+}&0&0&A_{+}^{2}u^*\\
0&w&A_{+}A_{-}y&0\\
0&A_{+}A_{-}y&w&0\\
A_{-}^{2}u&0&0&v_{-}\\
\end{pmatrix},
\end{align}
where
\begin{align}
  A_{\pm}&=\frac{Z\left(\beta,h\pm2\text{i}\eta t/\beta\right)}{Z\left(\beta,h\right)} \nonumber\\
  &=\frac{e^{2\text{i}N_b\eta t}\prod_{n=1}^{N_b}\left(e^{-2\beta \left(h\pm2\text{i}\eta t/\beta\right)}-z_{n}\right)}{\prod_{n=1}^{N_b}\left(e^{-2\beta h}-z_{n}\right)}\label{eq:Z_1}
\end{align}
is the analogous partition function.
If we consider the probe coupled to a ferromagnetic Ising bath under zero field ($h=0$), we can verify that the analogous partition function
\begin{align}
  A\left(\text{i}x\right)=\frac{e^{\text{i}\beta N_bx}\prod_{n=1}^{N_b}\left(\text{e}^{-2\text{i}\beta x}-\text{e}^{\text{i}\phi_n}\right)}{\prod_{n=1}^{N_b}\left(1-\text{e}^{\text{i}\phi_n}\right)}\label{A}
\end{align}
is a real even function (see Appendix \ref{Appendix_A}), where $x$ is real, implying the relationship $A^*\left(\text{i}x\right)=A\left(\text{i}x\right)=A\left(-\text{i}x\right)$.
By this way, we denote $A=A\left(\pm2\text{i}\eta t/\beta\right)$.

Therefore, we could rewrite the  evolution of the two-spin reduced density matrix in a block-diagonal form \cite{Wang2003}
\begin{align}
  \rho_{\Rmnum{1}}\left(t\right)=\begin{pmatrix}
v_{+}&A^{2}u^*\\
A^{2}u&v_{-}\\
\end{pmatrix}\oplus\begin{pmatrix}
w&A^{2}y\\
A^{2}y&w\\
\end{pmatrix},\label{rho_III}
\end{align}
in the basis $\left\{\left|11\right\rangle,\left|00\right\rangle,\left|01\right\rangle,\left|10\right\rangle\right\}$.
Comparing with the definition of the dephasing channel, one can easily find that $\rho_{\Rmnum{1}}\left(t\right)$ is quantum dephasing channels $\mathcal{E}_{\Rmnum{1}}\left(\rho_P\right)$, and the three Kraus operators are given by
\begin{gather}
 M_0\!=\!\sqrt{A^{2\!}}I, \!M_1\!=\!\sqrt{\!1\!-\!A^{2\!}}|0\rangle\langle0|,\!M_2\!=\!\sqrt{\!1\!-\!A^{2\!}}|1\rangle\langle1|,\label{Kraus1}
\end{gather}
where $I$ is the identity operator.
\subsection{Dephasing channel where $N$ probes are coupled to one bath together.}
We consider the dephasing channel where multiprobes are coupled to one $N_b$-spins bath together.
From the previous subsection, we can gain the total Hamiltonian and the unitary matrix
\begin{align}
H_{\Rmnum{2}}&=H_0\otimes I_{2N\times2N}+\eta H_1\otimes\sum^N_{k=1}\sigma_{k,z},\label{H_II}\\
U_{\Rmnum{2}}&= \text{e}^{-\text{i}H_{\Rmnum{2}}t}=\text{e}^{-\text{i}H_0t}\text{e}^{-\text{i}\eta \sum^N_{j=k}\sigma_{k,z}H_1t}.\label{U_II}
\end{align}

Here, we consider the two-qubit reduced state.
The density matrix of the general initial state is given by
\begin{align}
  \rho_{\Rmnum{2}}\left(0\right)=\left(\sum_{i,j=0}^{3}{\rho_{ij}\left|i\right\rangle\left\langle j\right|}\right)\otimes\rho_B\left(0\right),
\end{align}
in the basis $\left\{\left|11\right\rangle,\left|00\right\rangle,\left|01\right\rangle,\left|10\right\rangle\right\}$.
From Eqs.~\eqref{rho_t},~\eqref{eq:T11_1}-\eqref{eq:simplify_1} and~\eqref{H_II}-\eqref{U_II}, the reduced initial density matrix $\rho_P\left(0\right)$ evolves to
\begin{align}
  \rho_{\Rmnum{2}}\left(t\right)=\begin{pmatrix}
\rho_{11}&A\rho_{12}&A\rho_{13}&A^{\prime}\rho_{14}\\
A\rho_{21}&\rho_{22}&\rho_{23}&A\rho_{24}\\
A\rho_{31}&\rho_{32}&\rho_{33}&A\rho_{34}\\
A^{\prime}\rho_{41}&A\rho_{42}&A\rho_{43}&\rho_{44}\\
\end{pmatrix},
\end{align}
where $A^{\prime}=A\left(\pm4\text{i}\eta t/\beta\right)$ if we consider the system under zero field ($h=0$).
Because the relationship between $A$ and $A^{\prime}$ is uncertain, the evolution of a general initial state is not a quantum channel.
However, if we choose a standard one-axis twisted state as Eq.~\eqref{rho_ss}, the evolution of the reduced initial state is
\begin{equation}
\rho_{\Rmnum{2}}\left(t\right)=\begin{pmatrix}
v_{+}&0&0&A^{\prime}u^*\\
0&w&y&0\\
0&y&w&0\\
A^{\prime}u&0&0&v_{-}\\
\end{pmatrix}.\label{rho_II}
\end{equation}
We can find that $\rho_{\Rmnum{2}}\left(t\right)$  is a quantum dephasing channel $\mathcal{E}_P\left(\rho_{\Rmnum{2}}\right)$ and the Kraus operators are given by
\begin{gather}
    M_0=\sqrt{1-A^{\prime}}|00\rangle\langle00|, M_1=\sqrt{1-A^{\prime}}|11\rangle\langle11|,\nonumber\\
    M_{2}\!=\!\sqrt{\!A^{\prime}}\left(|00\rangle\!\langle00|\!+\!|11\rangle\!\langle11|\right)\!,\!M_{3}\!=\!|01\rangle\!\langle01|\!+\!|10\rangle\!\langle10|.\label{Kraus}
\end{gather}
\section{Multiprobes only coupled to their own bath}\label{Sec. V}
In the previous section, we have proposed the Lee-Yang dephasing channels and have proven that the probe(s)-bath system, whose ferromagnetic Ising bath is under zero field, is a Lee-Yang dephasing channel.
In this section, we calculate the coherence, the rescaled concurrence and the spin-squeezing parameter of multiprobes only respectively coupled to their own bath.
And we discuss the influence of the bath spins number $N_b$, the inverse temperature $\beta$, the probes number $N$ and the twisted angle $\theta$ on the them.

To take advantage of exchange symmetry, our initial state is one-axis twisted state shown as Eq.~\eqref{rho_ss}.
From the dephasing channels $\mathcal{E}_{\Rmnum{1}}\left(\rho_P\right)$ we showed in Sec.~\ref{Sec. IV},
the coherence of the system is
\begin{align}
  L\left(t\right)=A^2\left(\left|u\right|+\left|u^*\right|+2\left|y\right|\right)=2A^2\left(\left|u\right|+y\right).
\end{align}

Actually, it is difficult to show a tangible form of the time constant of the decay because the partition function $Z\left(\beta,h\right)$ cannot be gotten straightway.
Consequently, we choose the recovery time $t=\left(2n-1\right)T\left(A^2\right)/\left(2N_b\right)$ when the coherence commences the recovery from the minimum value.
The recovery time describes the strength of the decoherence and can be read from the figures.
To give a specific impression about our idea, we elaborate on it by the one-dimensional (1D) Ising model with nearest-neighbor ferromagnetic coupling $\lambda=1$, zero field ($h = 0$) and the periodic boundary condition.


\begin{figure}[htbp]
\begin{minipage}{0.49\linewidth}
\begin{overpic}[width=\linewidth,height=\linewidth]{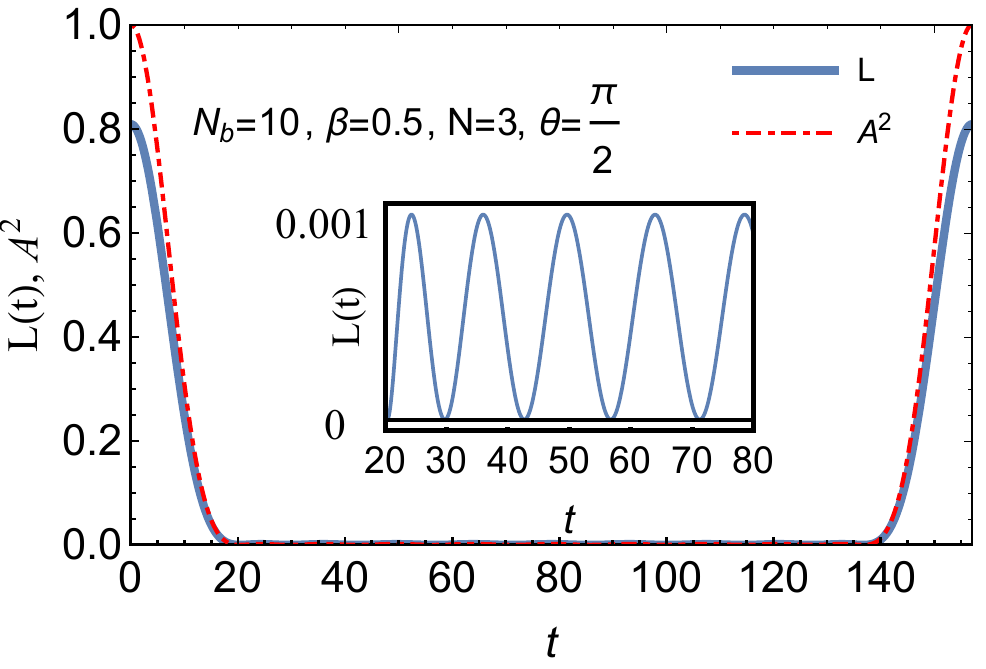}\label{fig2-a}
\put(0,102){$(a)$}
\end{overpic}
\end{minipage}
\begin{minipage}{0.49\linewidth}
\begin{overpic}[width=\linewidth,height=\linewidth]{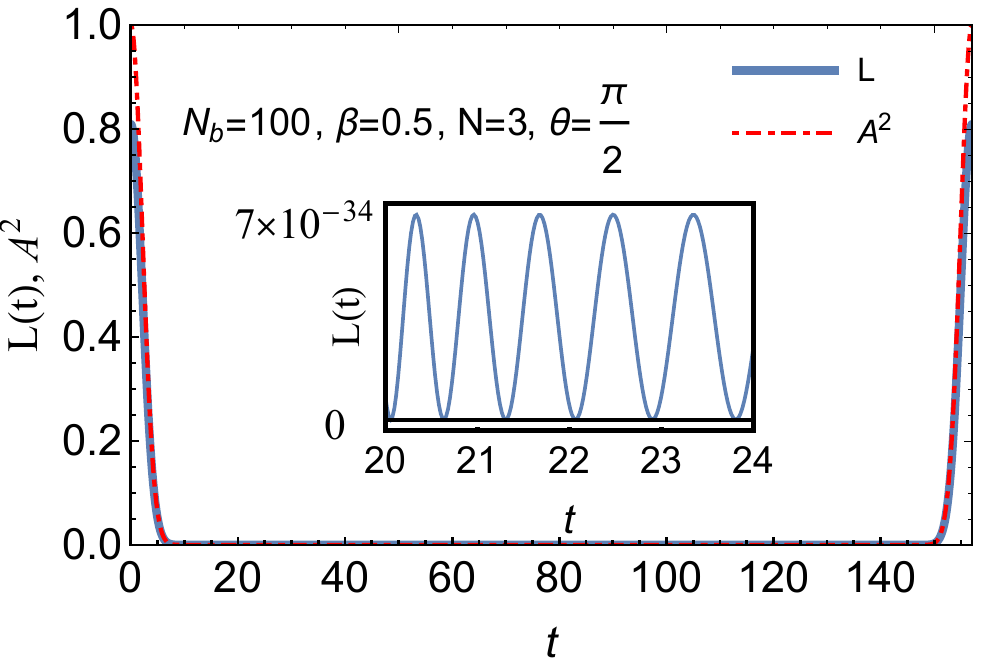}\label{fig2-b}
\put(0,102){$(b)$}
\end{overpic}
\end{minipage}
\begin{minipage}{0.49\linewidth}
\begin{overpic}[width=\linewidth,height=\linewidth]{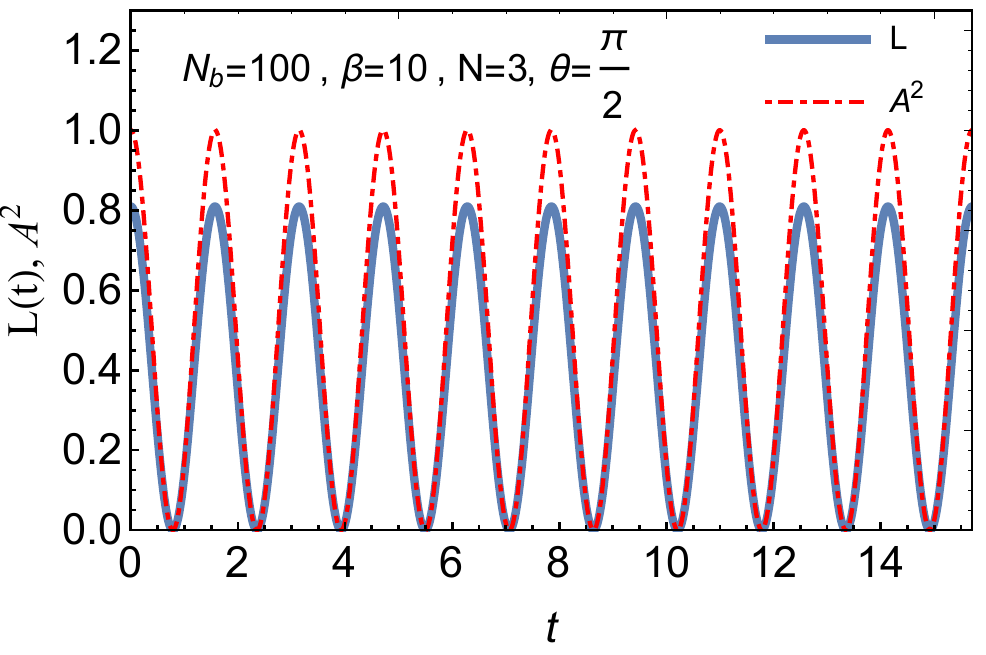}\label{fig2-c}
\put(0,102){$(c)$}
\end{overpic}
\end{minipage}
\begin{minipage}{0.49\linewidth}
\begin{overpic}[width=\linewidth,height=\linewidth]{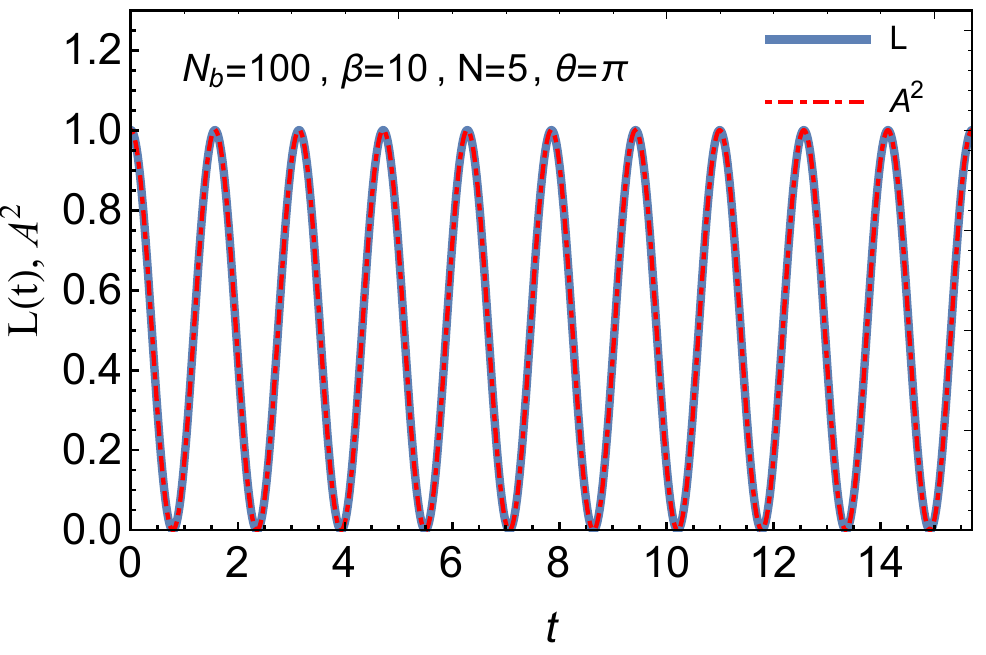}\label{fig2-d}
\put(0,102){$(d)$}
\end{overpic}
\end{minipage}
\protect\protect\caption{
Relationship between the Lee-Yang zeros and the coherence of the probes coupled to their own bath.
(a), (b) Coherence (blue line) and the analogous partition function $A^2$ (red dot dash line) for a 1D Ising model with the number of spins $N_b=10$ and $100$, respectively, for the inverse temperature $\beta=0.5$, the number of the probes $N=3$ and the twist angle $\theta=\pi/2$.
The coherence and the function $A^2$ with $N_b=100$ and $\beta=10$ for (c) $N=3$, $\theta=\pi/2$ and (d) $N=5$, $\theta=\pi$.
The insets in (a) and (b) zoom into the coherence zero.
The blue lines mark the coherences and the red dot dash lines mark the functions $A^2$.
The probes-bath coupling constant is $\eta=0.01$.}
 \label{fig2}
\end{figure}

Fig.~\ref{fig2} shows the Lee-Yang zeros and the coherence of the probes coupled to their own bath at various parameters.
At infinite temperature ($\beta=0$), all the Lee-Yang zeros are degenerate at $z_{n}=-1$ and we can easily find that the Lee-Yang zeros ($z_{n}$) only depend on the particle number ($N_b$) and the temperature ($\beta$ or $1/T$) .
As shown in Figs.~2(a) and 2(b), one can find that the coherence decreases to zero, and emerges $N_b$ peaks after the first Lee-Yang singularity.
With the bath size increasing toward the thermodynamic limit, the values of $N_b$ peaks decrease dramatically.
Obviously, the period of the coherence is $T\left(A^2\right)=2\pi/(4\eta)$.
From Eq.~\eqref{A}, supplied the uniformly distributed Lee-Yang zeros ($\phi_{n}\rightarrow\left(2n-1\right)\pi/N_b$, with $\beta\rightarrow\infty $), the coherence commences the recovery when the time satisfies $t=\left(2n-1\right)T\left(A^2\right)/\left(2N_b\right)$, where $n=1,2,3,\dots$ [Fig.~2(c)].
As shown in Fig.~2(d), the probes number $N$ and the twist angle $\theta$ only impact on the amplitude of the coherence.
Furthermore, one can find that times in correspondence to the Lee-Yang zeros are also the zeros of the coherence.

Now, we calculate the entanglement of the probe-bath systems.
As a prior knowledge, from Eqs.~\eqref{v_pm}-\eqref{sigma}, we know $\sqrt{v_+v_-}\geq\left|u\right|$, $v_\pm$ and $y=w>0$ are real.
Employing Eqs.~\eqref{R} and~\eqref{rho_III},
we immediately get the square roots of eigenvalues
\begin{align}
\lambda_{1,2}&=\sqrt{v_+v_-}\pm A^2\left|u\right|,\\
\lambda_{3,4}& =\left(1\pm A^2\right)y.
\end{align}
From Eq.~\eqref{eq:C}, we obtain the evolution of the concurrence as
\begin{align}
C\left(\rho\left(t\right)\right)=2\text{max}\left\{0,A^2\left|u\right|-y,A^2y-\sqrt{v_+v_-}\right\}.
\end{align}
From the result we have gained, the rescaled concurrence is
\begin{align}
  C_r=\text{max}\left\{0,A^2C_r\left(0\right)+2\left(N-1\right)\left(A^2-1\right)y\right\},
\end{align}
where $C_r\left(0\right)=2\left(N-1\right)\text{max}\left\{0,\left|u\right|-y\right\}$ is the initial concurrence.

To illustrate the above result, we test the one-dimensional (1D) Ising model with nearest-neighbor ferromagnetic coupling $\lambda=1$, zero field ($h = 0$) and the periodic boundary condition.

\begin{figure}[htbp]
\centering
\begin{minipage}{0.49\linewidth}
\begin{overpic}[width=\linewidth,height=\linewidth]{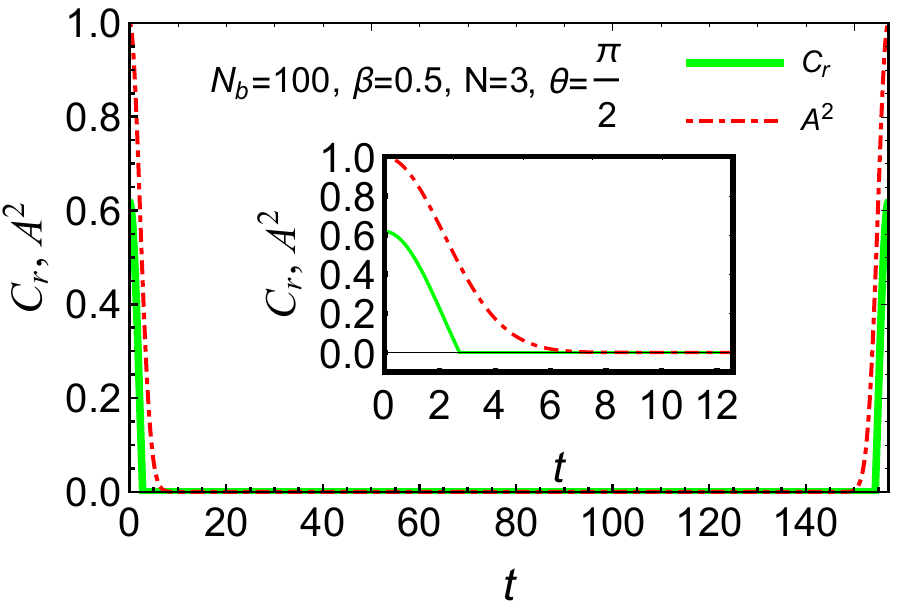}\label{fig3-a}
\put(0,102){$(a)$}
\end{overpic}
\end{minipage}
\begin{minipage}{0.49\linewidth}
\begin{overpic}[width=\linewidth,height=\linewidth]{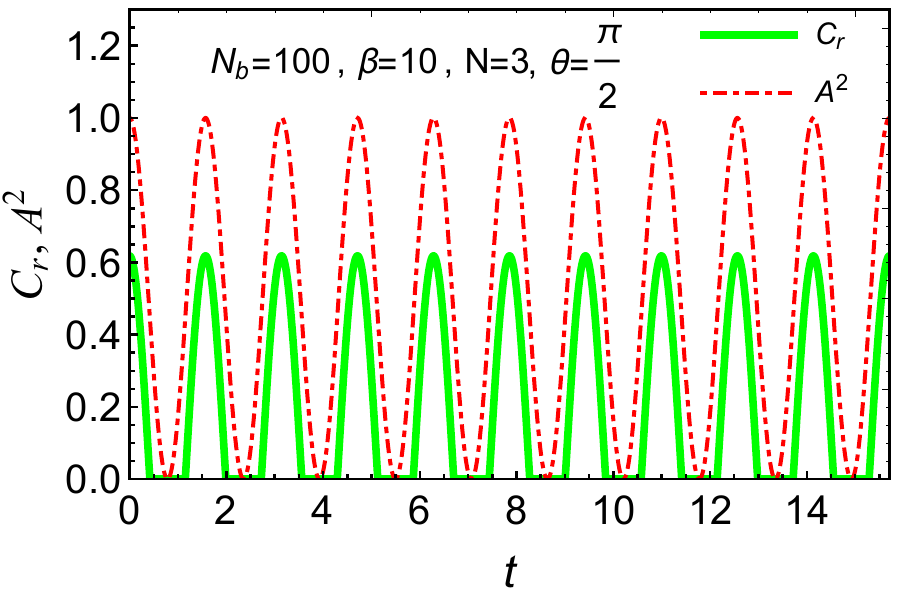}\label{fig3-b}
\put(0,102){$(b)$}
\end{overpic}
\end{minipage}
\begin{minipage}{0.49\linewidth}
\begin{overpic}[width=\linewidth,height=\linewidth]{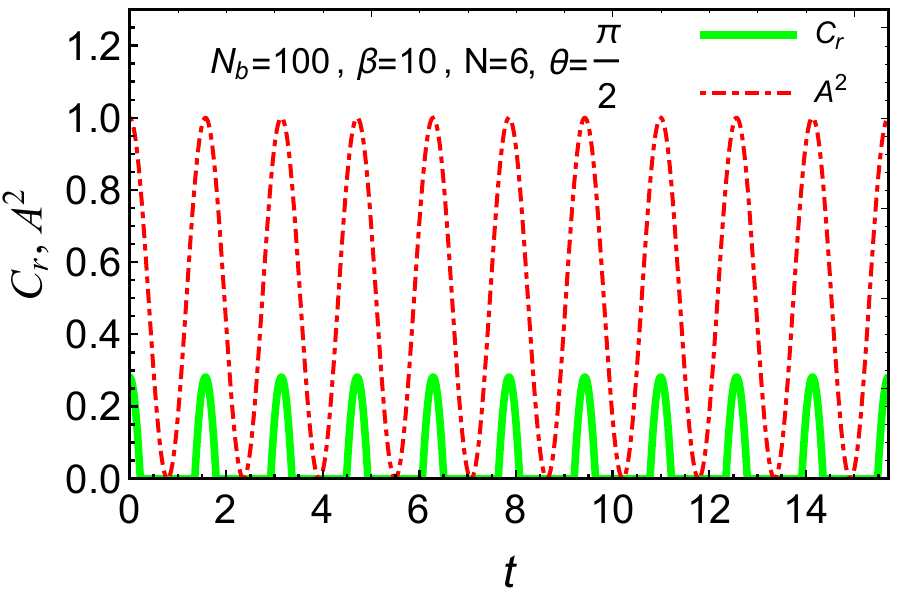}\label{fig3-c}
\put(0,102){$(c)$}
\end{overpic}
\end{minipage}
\begin{minipage}{0.49\linewidth}
\begin{overpic}[width=\linewidth,height=\linewidth]{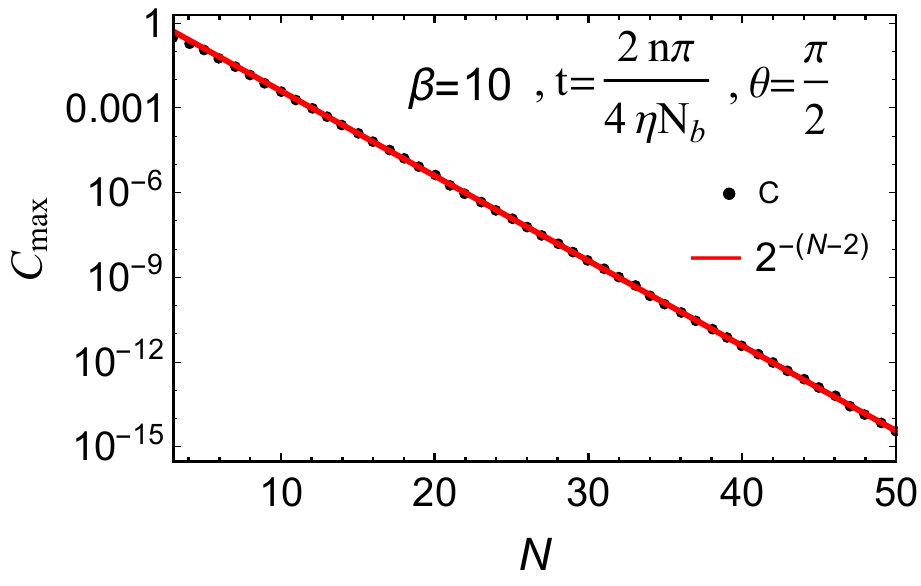}\label{fig3-d}
\put(0,102){$(d)$}
\end{overpic}
\end{minipage}
\protect\protect\caption{
Relationship between the Lee-Yang zeros and the rescaled concurrence of the probes coupled to their own bath with various inverse temperatures $\beta$ and the number of probes $N$.
The rescaled concurrence (green line) and the analogous partition function $A^2$ (red dot dash line) for (a) $\beta=0.5$, $N=3$, (b) $\beta=10$, $N=3$ and (c) $\beta=10$, $N=6$.
(d) Maximum original concurrence (black dots) and the fitting function (red line) at the inverse temperature $\beta=10$.
The probes-bath coupling is $\eta=0.01$, the number of the bath spins $N_b=100$ and the twist angle $\theta=\pi/2$.
}
\label{fig3}
\end{figure}
Fig.~\ref{fig3} shows the rescaled concurrence of the probes coupled to their own bath with probes-bath coupling $\eta=0.01$, $N_b=100$ bath spins and the twist angle $\theta=\pi/2$ at various temperatures $\beta$ and the number of probes $N$.
As seen in Fig.~3(a), under the thermodynamic limit, the rescaled concurrence decreases rapidly to zero with proper temperature $\beta$.
However, with low temperature, the rescaled concurrence emerges peaks at $t=nT\left(A^2\right)/N_b$, where $n=1,2,3,\dots$ [Fig.~3(b)].
The probes number $N$ and the twist angle $\theta$ only impact on the amplitude of the rescaled concurrence [Fig.~3(c)].
Times in correspondence to the Lee-Yang zeros are the centers of all the vanishing domains of the rescaled concurrence at low temperature.
As shown in Fig.~3(d), given the twist angle $\theta$, the maximum original concurrence is $C_{\text{max}}\left(\rho\right)= \exp\left[\alpha \left(N-2\right)\right]$, where $\alpha$ is a constant only depended on $\theta$.
With $\theta=\pi/2$ and low temperature, the maximum original concurrence is $C_{\text{max}}\left(\rho\right)= 2^{-\left(N-2\right)}$ at times satisfying $t=nT\left(A^2\right)/N_b$.

From Eqs.~\eqref{QC} and~\eqref{Kraus1}, the evolution of the matrix $O$ under the dephasing channel is
\begin{align}
\mathcal{E}_{\Rmnum{1}}\left(O\right)=\mathcal{E}_{\Rmnum{1}}^{\dagger}\left(O\right)=\begin{pmatrix}
O_{00}&A^2O_{01}\\
A^2O_{10} & O_{11}
\end{pmatrix},
\end{align}
and we find that
\begin{align}
  \mathcal{E}_{\Rmnum{1}}^{\dagger}\left(\sigma_z\right)&=\sigma_z\\
  \mathcal{E}_{\Rmnum{1}}^{\dagger}\left(\sigma_i\right)&=A^2\sigma_i\quad\text{for}\quad i=x,y.\label{E1}
\end{align}
From Eqs.~\eqref{O} and~\eqref{E1}, one have
\begin{align}
\left\langle\sigma_{1+}\sigma_{2-}\right\rangle&=A^{2}\left\langle\sigma_{1+}\sigma_{2-}\right\rangle_0,\\
\left\langle\sigma_{1-}\sigma_{2-}\right\rangle&=A^{2}\left\langle\sigma_{1-}\sigma_{2-}\right\rangle_0,
\end{align}
where $\left\langle O\right\rangle=\left\langle\psi\left(0\right)\right| O\left|\psi\left(0\right)\right\rangle$.

Since each probe is only coupled to its own bath, the independent and identical dephasing channels act separately on each probe spin.
Consequently, the spin-squeezing parameters are obtained as
\begin{align}
\xi^2&=1-2\left(N-1\right)A^2\left(\left|u\right|-y\right)=1-A^2C_r\left(0\right),\\
\xi^{\prime 2}&=\!1\!-\!\text{max}\left\{\!0,A^2C_r\!\left(0\right)\!+\!2\!\left(N\!-\!1\right)\!\left(A^2\!-\!1\right)\!y\right\}.
\end{align}
Therefore, under the dephasing channels, the spin-squeezing parameter is unbound by the relationship $\xi^2=1-C_r$ and the performance of spin squeezing is improved by
\begin{align}
\Delta\xi^2&\equiv \xi^{\prime 2}-\xi^2 \nonumber\\
&=2\left(N-1\right)\text{min}\left\{A^2\left(\left|u\right|-y\right),\left(1-A^2\right)y\right\}.
\end{align}
The maximum of the improvement
\begin{align}
\Delta\xi^2_{\text{max}}=2\left(N-1\right)\left(1-\frac{y}{\left|u\right|}\right)y
\end{align}
only depends on its initial state.

\begin{figure}[htbp]
\centering
\begin{minipage}{0.49\linewidth}
\begin{overpic}[width=\linewidth,height=\linewidth]{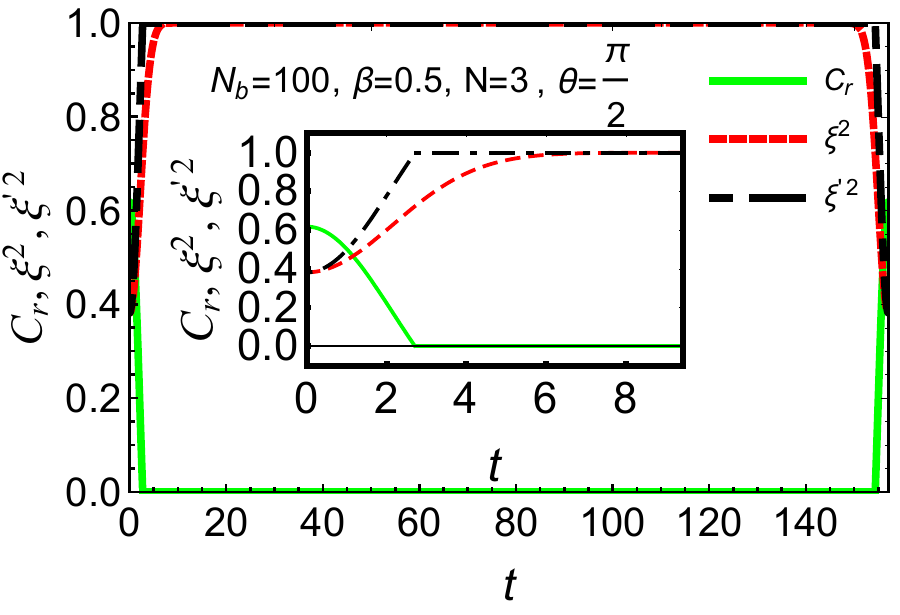}\label{fig4-a}
\put(0,102){$(a)$}
\end{overpic}
\end{minipage}
\begin{minipage}{0.49\linewidth}
\begin{overpic}[width=\linewidth,height=\linewidth]{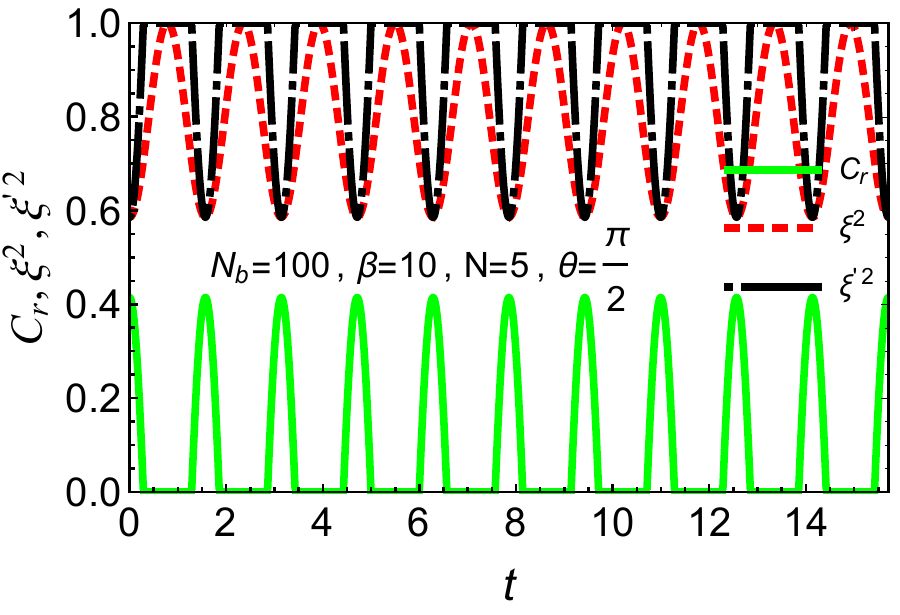}\label{fig4-b}
\put(0,102){$(b)$}
\end{overpic}
\end{minipage}
\protect\protect\caption{
Relationship between the rescaled concurrence and the spin-squeezing parameters of the probes coupled to their own bath with various inverse temperatures $\beta$ and the number of probes $N$.
The rescaled concurrence (green line) and the spin-squeezing parameters $\xi^2$ (red dotted line), $\xi^{\prime 2}$ (black dot dash line) for (a) $\beta=0.5$, $N=3$, (b) $\beta=10$, $N=5$.
The inset in (a) zooms into the initial time.
The probes-bath coupling is $\eta=0.01$, the number of the bath spins $N_b=100$ and the twist angle $\theta=\pi/2$.
}
\label{fig4}
\end{figure}
Fig.~\ref{fig4} shows the rescaled concurrence and the spin-squeezing parameters of the probes coupled to their own bath with probes-bath coupling $\eta=0.01$, $N_b=100$ bath spins and the twist angle $\theta=\pi/2$ at various temperatures $\beta$ and the number of probes $N$.
As seen in Fig.~4(a), the spin-squeezing parameter $\xi^2$ is always smaller then the corresponding $\xi^{\prime 2}$, which means it does not hold the relationship $\xi^2=1-C_r$ under decoherence.
Moreover, with temperature decreasing, the spin-squeezing parameter is recovered periodically at time satisfying $T\left(A^2\right)/N_b$ [Fig.~4(b)].

\section{Multiprobes Coupled to one Bath Together}\label{Sec. VI}
In this section, we obtain and analyse the coherence, the rescaled concurrence and the spin-squeezing parameter of probes coupled to a mutual many-body system.
We compare them of both systems we have mentioned.

To take advantage of exchange symmetry, our initial state is one-axis twisted state shown as Eq.~\eqref{rho_ss}.
From the dephasing channel $\mathcal{E}_{\Rmnum{2}}\left(\rho_P\right)$ we showed in Sec.~\ref{Sec. IV}, one can gain the coherence
\begin{align}
L\left(t\right)=2\left(\left|A^{\prime}u\right|+y\right).
\end{align}

Employing Eqs.~\eqref{R} and~\eqref{rho_II}, we obtain the square roots of eigenvalues as follows:
\begin{align}
\lambda_{1,2}&=\sqrt{v_+v_-}\pm \left|A^{\prime}u\right|,\\
\lambda_{3}& =2y,\\
\lambda_4&=0.
\end{align}
From Eq.~\eqref{eq:C}, we obtain the evolution of the rescaled concurrence as
\begin{align}
C_r&=2\left(N-1\right)\text{max}\left\{0,\left|A^{\prime}u\right|-y,y-\sqrt{v_+v_-}\right\}\nonumber\\
&=\text{max}\left\{0,\left|A^{\prime}\right|C_r\left(0\right)+2\left(N-1\right)\left(\left|A^{\prime}\right|-1\right)y\right\},
\end{align}
where $C_r\left(0\right)=2\left(N-1\right)\text{max}\left\{0,\left|u\right|-y\right\}$ is the initial concurrence.

From Eqs.~\eqref{QC},~\eqref{O} and~\eqref{Kraus}, the evolution of the matrix $O$ under the dephasing channel is
\begin{align}
\mathcal{E}_{\Rmnum{2}}\left(O\right)=\mathcal{E}_{\Rmnum{2}}^{\dagger}\left(O\right)=\begin{pmatrix}
O_{11}&0&0&A^{\prime}O_{14}\\
0&O_{22}&O_{23}&0\\
0&O_{32}&O_{33}&0\\
A^{\prime}O_{41}&0&0&O_{44}\\
\end{pmatrix},
\end{align}
and we find that
\begin{align}
\left\langle\sigma_{1+}\sigma_{2-}\right\rangle&=\left\langle\sigma_{1+}\sigma_{2-}\right\rangle_0,\\
\left\langle\sigma_{1-}\sigma_{2-}\right\rangle&=A^{\prime}\left\langle\sigma_{1-}\sigma_{2-}\right\rangle_0.
\end{align}
Then we obtain the evolution of the spin-squeezing parameter
\begin{align}
\xi^2=\xi^{\prime 2}=1-C_r=1-2\left(N-1\right)\left(\left|A^{\prime}u\right|-y\right).
\end{align}
That means decoherence does not destroy the relationship $\xi^2+C_r=1$ during the whole evolution.
Analytical results for time-evolution of all relevant coherences, rescaled concurrences, expectations and spin-squeezing parameters, chosen a one-axis twisted state as the initial state, are given in Table.~\ref{Table}.

\begin{table*}[htbp]
	\centering
	\caption{Analytical results for time-evolution of all relevant coherences, rescaled concurrences, expectations and spin-squeezing parameters, chosen the one-axis twisted state as the initial state.}
	\begin{tabular}{cp{7cm}<{\centering}p{7cm}<{\centering}}
		\toprule[1pt]
\specialrule{0em}{2pt}{2pt}
		&Dephasing channels &Dephasing channel \\
&where multiprobes are coupled to their own bath&where multiprobes are coupled to one bath together\\
\specialrule{0em}{2pt}{2pt}	
\midrule[0.5pt]
\specialrule{0em}{2pt}{2pt}
		$L\left(t\right)$&$2A^2\left(\left|u\right|+y\right)$&$2\left(\left|A^{\prime}u\right|+y\right)$\\
\specialrule{0em}{2pt}{2pt}
$Cr$&$\text{max}\left\{0,A^2C_r\left(0\right)+2\left(N-1\right)\left(A^2-1\right)y\right\}$&$\text{max}\left\{0,\left|A^{\prime}\right|C_r\left(0\right)+2\left(N-1\right)\left(\left|A^{\prime}\right|-1\right)y\right\}$\\
\specialrule{0em}{2pt}{2pt}
$\langle\sigma_{1z}\sigma_{2z}\rangle$&$\langle\sigma_{1z}\sigma_{2z}\rangle_0$&$\langle\sigma_{1z}\sigma_{2z}\rangle_0$\\
\specialrule{0em}{2pt}{2pt}
$\langle\sigma_{1+}\sigma_{2-}\rangle$&$\langle\sigma_{1+}\sigma_{2-}\rangle_0$&$\langle\sigma_{1+}\sigma_{2-}\rangle_0$\\
\specialrule{0em}{2pt}{2pt}
$\langle\sigma_{1-}\sigma_{2-}\rangle$&$A\langle\sigma_{1-}\sigma_{2-}\rangle_0$&$A^{\prime}\langle\sigma_{1-}\sigma_{2-}\rangle_0$\\
\specialrule{0em}{2pt}{2pt}
$\xi^2$&$1-A^2C_r\left(0\right)$ &$1-C_r$\\
\specialrule{0em}{2pt}{2pt}
		\bottomrule[1pt]
	\end{tabular}\label{Table}
\end{table*}

\begin{figure}[htbp]
\centering
\begin{minipage}{0.49\linewidth}
\begin{overpic}[width=\linewidth,height=\linewidth]{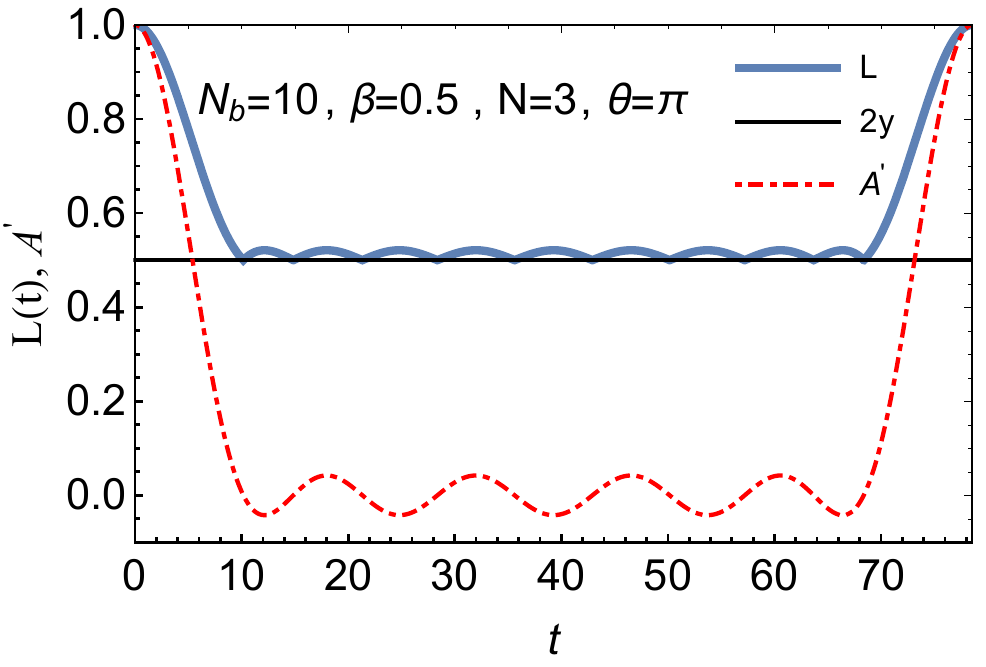}\label{fig4-a}
\put(0,102){$(a)$}
\end{overpic}
\end{minipage}
\begin{minipage}{0.49\linewidth}
\begin{overpic}[width=\linewidth,height=\linewidth]{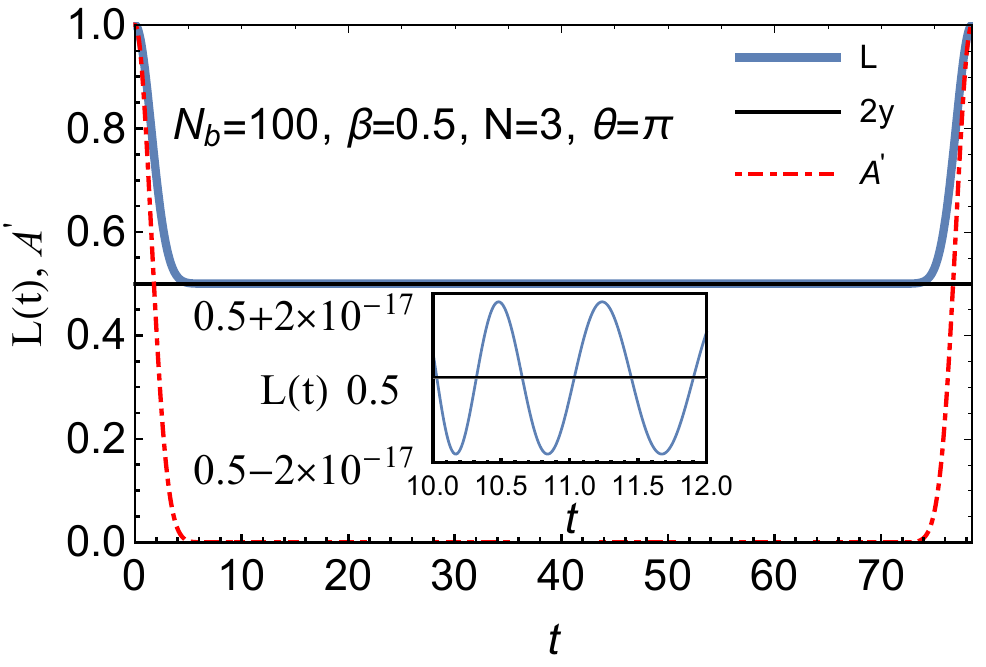}\label{fig4-b}
\put(0,102){$(b)$}
\end{overpic}
\end{minipage}
\begin{minipage}{0.49\linewidth}
\begin{overpic}[width=\linewidth,height=\linewidth]{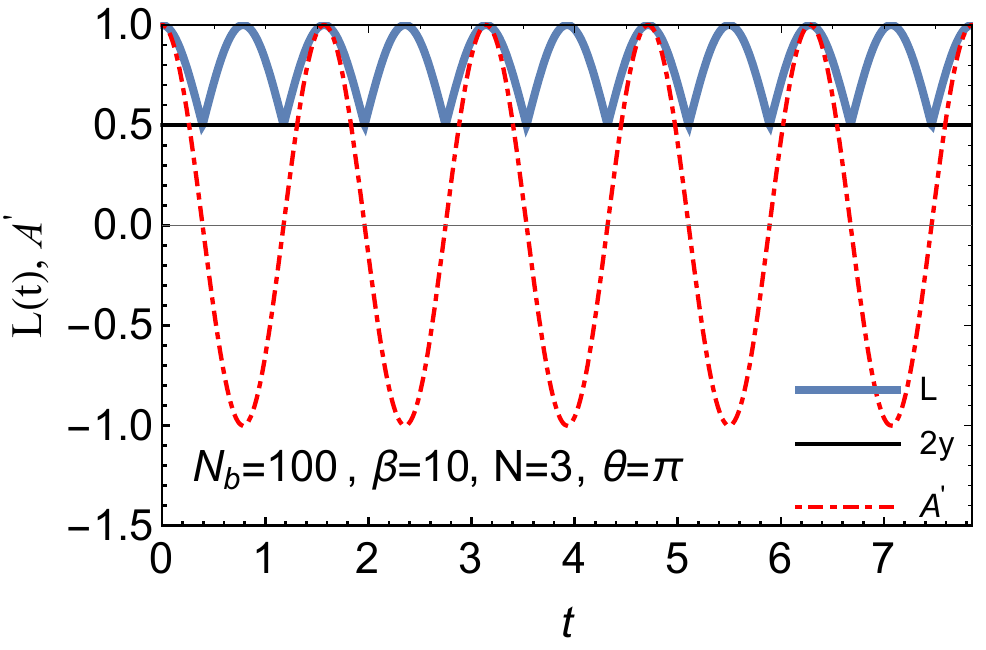}\label{fig4-c}
\put(0,102){$(c)$}
\end{overpic}
\end{minipage}
\begin{minipage}{0.49\linewidth}
\begin{overpic}[width=\linewidth,height=\linewidth]{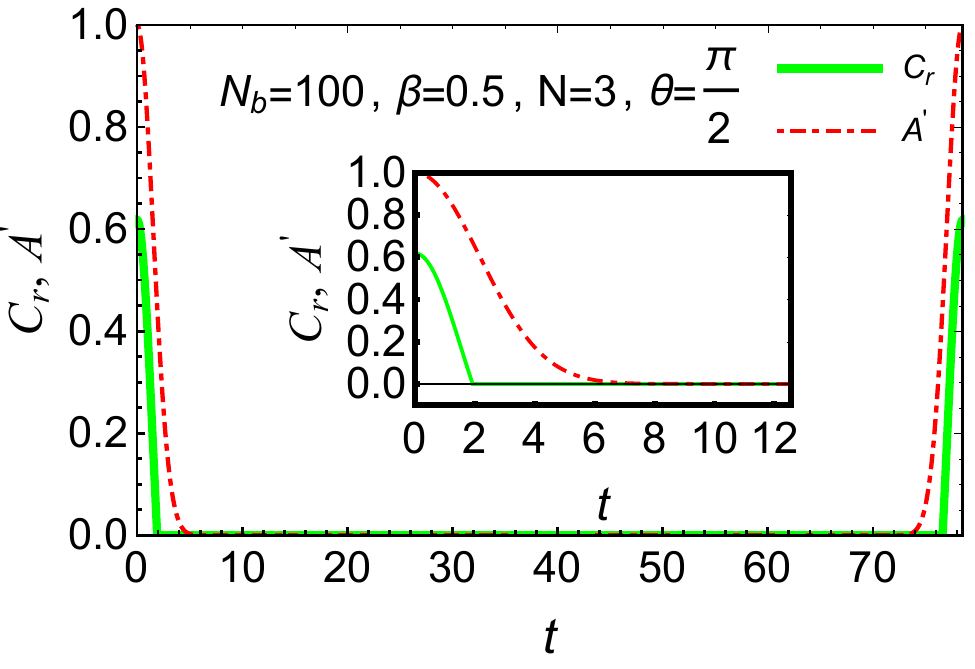}\label{fig4-d}
\put(0,102){$(d)$}
\end{overpic}
\end{minipage}
\begin{minipage}{0.49\linewidth}
\begin{overpic}[width=\linewidth,height=\linewidth]{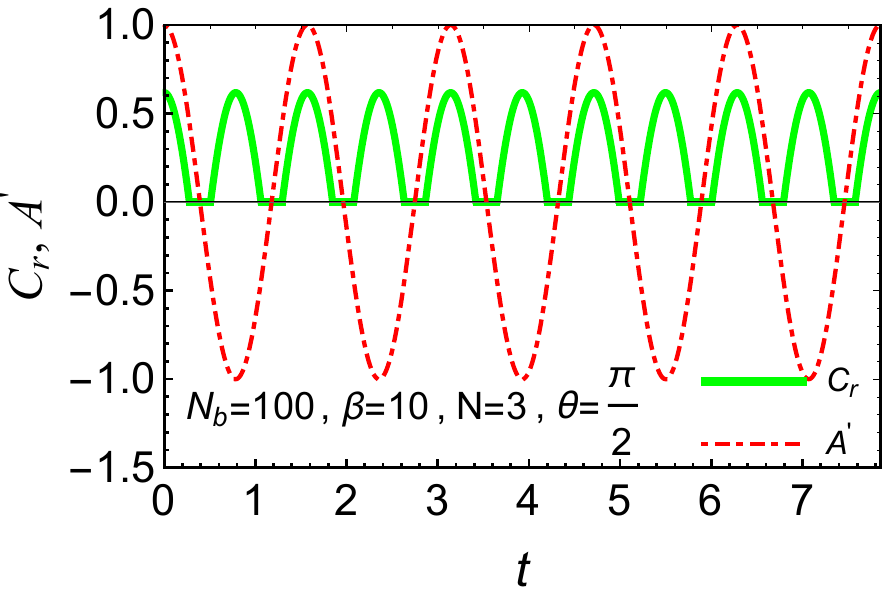}\label{fig4-e}
\put(0,102){$(e)$}
\end{overpic}
\end{minipage}
\begin{minipage}{0.49\linewidth}
\begin{overpic}[width=\linewidth,height=\linewidth]{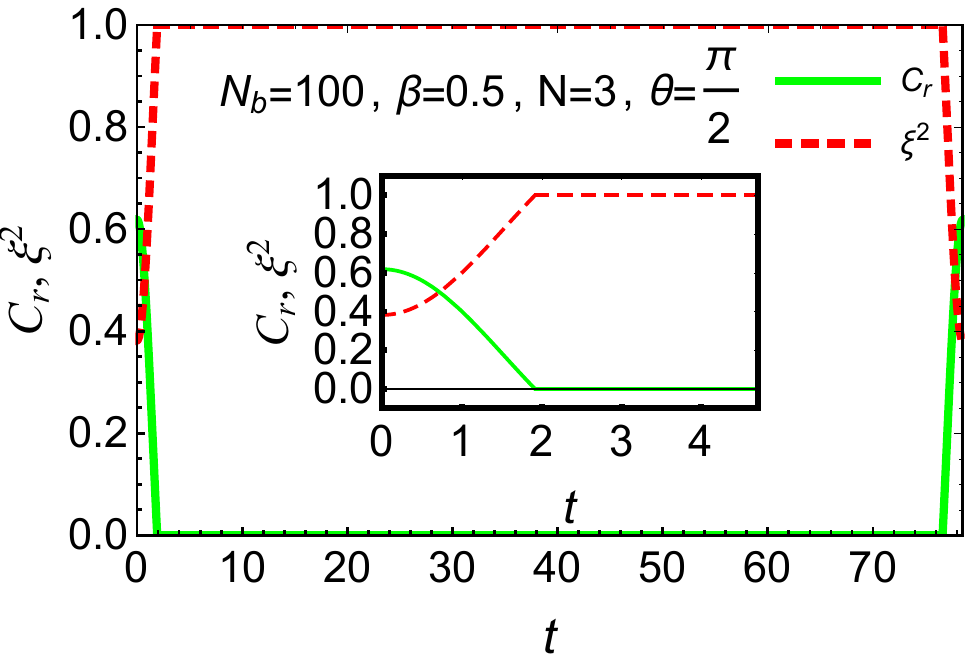}\label{fig4-f}
\put(0,102){$(f)$}
\end{overpic}
\end{minipage}
\protect\protect\caption{
Coherence, rescaled concurrence and spin-squeezing parameters of the probes coupled to one bath together.
The coherences (blue line) and the analogous partition functions $A^{\prime}$ (red dot dash line) with the twist angle $\theta=\pi$ for (a) $N_b=10$, $\beta=0.5$, (b) $N_b=100$, $\beta=0.5$ and (c) $N_b=100$, $\beta=10$.
The rescaled concurrences (green line) and the analogous partition functions $A^{\prime}$ (red dot dash line) with $N_b=100$ and $\theta=\pi/2$ for (d) $\beta=0.5$ and (e) $\beta=10$.
(f) The rescaled concurrences (green line) and the spin-squeezing parameters $\xi^2$ (red dotted line) for $N_b=100$, $\beta=0.5$ and $\theta=\pi/2$.
The insets of (b), (d) and (f) zoom into the coherence zero and the initial time, respectively.
The probes-bath coupling is $\eta=0.01$ and the number of probes $N=3$.}
\label{fig5}
\end{figure}

We show our results by the one-dimensional (1D) Ising model with nearest-neighbor ferromagnetic coupling $\lambda\!=\!1$, zero field ($h\!=\!0$) and the periodic boundary condition.
Fig.~\ref{fig5} shows the coherence, the rescaled concurrence and the spin-squeezing parameter of the probes coupled to one bath together.
As shown in Figs.~5(a) and~5(b), with appropriate temperature, one can find that the coherence is asymptotic to a stable value $\!\left(1\!-\!\cos^{N-2}\theta\right)\!/4$ because the analogous partition function $A^{\prime}$ becomes almost zero after the first Lee-Yang singularity.
With the bath size increasing toward the thermodynamic limit, one can gain the stable minimum value of the coherence in almost the whole time domain.
From Eq.~\eqref{A}, when temperature is low enough ($\beta\!\rightarrow\!\infty\!$), the coherence commences the recovery from the minimum value when the time satisfies $t\!=\!\left(2n\!-\!1\right)\!T\!\left(A^{\prime}\right)\!/\!\left(\!2N_b\!\right)$, where $\!T\!\left(A^{\prime}\right)\!=\!T\!\left(A^{2}\right)\!/2$ and $n=1,2,3,\dots$ [Fig.~5(c)].
As seen in Figs.~5(d) and~5(e), the rescaled concurrence shares the same properties with the correspondence in Sec.~\ref{Sec. V}, except for the double period.
Times in correspondence to the Lee-Yang zeros are the centers of all the vanishing domains of the rescaled concurrence at low temperature.
As shown in Fig.~5(f), one can find that the relationship $\xi^2\!=\!1\!-\!C_r$ is preserved  under decoherence.

\section{Conclusion and discussion}\label{Sec. VII}
In this work, we propose two different types of Lee-Yang dephasing channels: (a) probes coupled to their own bath and (b) probes coupled to one bath together.
We prove that the probe(s)-bath system whose ferromagnetic Ising bath is under zero field is a Lee-Yang dephasing channel.
Under the Lee-Yang dephasing channels, we obtain the coherence, the concurrence and the spin-squeezing parameter, with nearest-neighbor ferromagnetic coupling $\lambda\!=\!1$ and the periodic boundary condition.
In the first place, one can find that the coherence decreases to zero (minimum, under the latter channel) and emerges certain peaks after the first Lee-Yang singularity, of which values decrease dramatically with the bath size increasing toward the thermodynamic limit.
The coherence could commence a recovery at certain discrete times if the temperature is low enough.
Furthermore, one can find that times in correspondence to the Lee-Yang zeros are also the zeros of the coherence.
Secondly, the concurrence shares almost the same properties in both channels.
The centers of its vanishing domains are corresponding to the Lee-Yang zeros.
Besides, at the given twist angle, the maximum original concurrence only depends on the number of probes.
Finally, we find that the performance of spin squeezing is improved and its maximum only depends on the initial state under the first dephasing channels.
However, the corresponding performance is unimproved under the second dephasing channel.

All we have discussed above are in the ferromagnetic Ising models.
For other systems (e.g., antiferromagnetic Ising models), the Lee-Yang zeros may not lie on a unit circle.
However, according to Eq.~\eqref{eq:Z_1}, one can employ an external field $h$ and gain all the zeros of modulus $\exp\left(-2\beta h\right)$.
There is a  one-to-one mapping between the zeros (centers of vanishing domains) of coherence (concurrence) and the Lee-Yang zeros as the temperature approaches zero.
Through measuring the quantum coherence and concurrence, one can gain the information about the corresponding Lee-Yang zeros.
With the Lee-Yang zeros determined, the partition function of an intriguing many-body system can be rebuilt and one can obtain all the properties of the system.
Since all physical properties can be acquired, our results provide a new method to investigate many-body physics and extend a new perspective of the relationship between the entanglement and spin squeezing in probes-bath systems.

\section{Acknowledgments}

This work was supported by the National Key Research and Development Program of China (No.~2017YFA0304202 and No.~2017YFA0205700), the NSFC (Grants No.~11875231 and No.~11935012),
and the Fundamental Research Funds for the Central Universities through
Grant No.~2018FZA3005.

\begin{appendix}
\section{Proof of the property of the partition function}\label{Appendix_A}
Here, we prove that the partition function is a real and even function in a ferromagnetic Ising bath under zero field ($h\! =\! 0$).
To achieve our goal, we first give the following lemma.
\begin{lemma}
Under the periodic boundary condition, if $z=\exp{\left(\text{i}\phi\right)}$ is a zero of the partition function $Z\left(\beta,h\right)$, its complex conjugate $z^{*}=\exp{\left(-\text{i}\phi\right)}$ is also a zero.
\end{lemma}
\begin{proof}
We first consider $N_b$ is even, since the periodic boundary condition $f_n=f_{N_b-n}$, we gain
\begin{align}
\sum_{n=0}^{N_b}{f_nz^n}&=\sum_{n=0}^{\frac{N_b}{2}-1}{f_n\left(z^n+z^{N_b-n}\right)+f_{\frac{N_b}{2}}z^{\frac{N_b}{2}}}\nonumber\\
&=\sum_{n=0}^{\frac{N_b}{2}-1}{f_n\left(\text{e}^{n\text{i}\phi_n}+\text{e}^{\left(N_b-n\right)\text{i}\phi_n}\right)+f_{\frac{N_b}{2}}\text{e}^{\frac{N_b}{2}\text{i}\phi_n}}\nonumber\\
&=0.
\end{align}
Multiplying both sides of the above equation by the factor $\exp{\left(-N_b\text{i}\phi_n\right)}$, we have
\begin{align}
  \sum_{n=0}^{\frac{N_b}{2}-1}{f_n\left(\text{e}^{\left(N_b-n\right)\left(-\text{i}\phi_n\right)}+\text{e}^{n\left(-\text{i}\phi_n\right)}\right)+f_{\frac{N_b}{2}}\text{e}^{\frac{N_b}{2}\left(-\text{i}\phi_n\right)}}=0.
\end{align}
Therefore, $z^{*}=\exp{\left(-\text{i}\phi_n\right)}$ are the roots of the partition function $Z\left(\beta,h\right)$.

Then, we consider $N_b$ is odd, we have
\begin{align}
\sum_{n=0}^{N_b}{f_nz^n}&=\sum_{n=0}^{\frac{N_b-1}{2}}{f_n\left(z^n+z^{N_b-n}\right)}\nonumber\\
&=\sum_{n=0}^{\frac{N_b-1}{2}}{f_n\left(\text{e}^{n\text{i}\phi_n}+\text{e}^{\left(N_b-n\right)\text{i}\phi_n}\right)}\nonumber\\
&=0.
\end{align}
Employing the same skill, we can find that $z^{*}=\exp{\left(-\text{i}\phi_n\right)}$ are the roots of the partition function $Z\left(\beta,h\right)$.
\end{proof}

It explains the fact that distribution of  the Lee-Yang zeros are symmetric about the real number axis.
From the Lemma, we have the following proposition.
\begin{proposition}
Under the periodic boundary condition, the partition function $Z\left(\beta,\text{i}x\right)$ is a real and even function implying the relationship $Z^*\left(\beta,\text{i}x\right)=Z\left(\beta,\text{i}x\right)=Z\left(\beta,-\text{i}x\right)$, where $x$ is real.
\end{proposition}
\begin{proof}
From the Lemma, we know that
\begin{align}
\prod_{n=1}^{N_b}{\left(a+b\text{e}^{\text{i}\phi_n}\right)}=\prod_{n=1}^{N_b}{\left(a+b\text{e}^{-\text{i}\phi_n}\right)},\label{iphi_n}
\end{align}
where $a$ and $b$ are arbitrary complex numbers.
We calculate the ratio of them
\begin{align}
\frac{Z^*\left(\beta,\text{i}x\right)}{Z\left(\beta,\text{i}x\right)}&=\text{e}^{-2\text{i}\beta N_bx}\prod_{n=1}^{N_b}{\frac{\text{e}^{2\text{i}\beta x}-\text{e}^{-\text{i}\phi_n}}{\text{e}^{-2\text{i}\beta x}-\text{e}^{\text{i}\phi_n}}}\nonumber\\
&=\prod_{n=1}^{N_b}{\frac{1-\text{e}^{-2\text{i}\beta x-\text{i}\phi_n}}{\text{e}^{-2\text{i}\beta x}-\text{e}^{\text{i}\phi_n}}}\nonumber\\
&=\prod_{n=1}^{N_b}{\left(-\text{e}^{-\text{i}\phi_n}\right)}.
\end{align}
From Eq.~\eqref{iphi_n}, we have
\begin{align}
\frac{Z\left(\beta,\text{i}x\right)}{Z^*\left(\beta,\text{i}x\right)}&=\prod_{n=1}^{N_b}{\frac{1-\text{e}^{2\text{i}\beta x+\text{i}\phi_n}}{\text{e}^{2\text{i}\beta x}-\text{e}^{-\text{i}\phi_n}}}\nonumber\\
&=\prod_{n=1}^{N_b}{\left(-\text{e}^{\text{i}\phi_n}\right)}\nonumber\\
&=\prod_{n=1}^{N_b}{\left(-\text{e}^{-\text{i}\phi_n}\right)},
\end{align}
where $a=0$ and $b=-1$.
Therefore, we find that $\frac{Z^*\left(\beta,\text{i}x\right)}{Z\left(\beta,\text{i}x\right)}=\frac{Z\left(\beta,\text{i}x\right)}{Z^*\left(\beta,\text{i}x\right)}$, which means $Z^*\left(\beta,\text{i}x\right)=Z\left(\beta,\text{i}x\right)$.
Employing the same skill, we can also see that $\frac{Z\left(\beta,-\text{i}x\right)}{Z\left(\beta,\text{i}x\right)}=\frac{Z\left(\beta,\text{i}x\right)}{Z\left(\beta,-\text{i}x\right)}$, and the partition function $Z\left(\beta,\text{i}x\right)$ is a real and even function, where $x$ is real.
\end{proof}
\end{appendix}



\begin{thebibliography}{0}%
\makeatletter
\providecommand \@ifxundefined [1]{%
 \@ifx{#1\undefined}
}%
\providecommand \@ifnum [1]{%
 \ifnum #1\expandafter \@firstoftwo
 \else \expandafter \@secondoftwo
 \fi
}%
\providecommand \@ifx [1]{%
 \ifx #1\expandafter \@firstoftwo
 \else \expandafter \@secondoftwo
 \fi
}%
\providecommand \natexlab [1]{#1}%
\providecommand \enquote  [1]{``#1''}%
\providecommand \bibnamefont  [1]{#1}%
\providecommand \bibfnamefont [1]{#1}%
\providecommand \citenamefont [1]{#1}%
\providecommand \href@noop [0]{\@secondoftwo}%
\providecommand \href [0]{\begingroup \@sanitize@url \@href}%
\providecommand \@href[1]{\@@startlink{#1}\@@href}%
\providecommand \@@href[1]{\endgroup#1\@@endlink}%
\providecommand \@sanitize@url [0]{\catcode `\\12\catcode `\$12\catcode
  `\&12\catcode `\#12\catcode `\^12\catcode `\_12\catcode `\%12\relax}%
\providecommand \@@startlink[1]{}%
\providecommand \@@endlink[0]{}%
\providecommand \url  [0]{\begingroup\@sanitize@url \@url }%
\providecommand \@url [1]{\endgroup\@href {#1}{\urlprefix }}%
\providecommand \urlprefix  [0]{URL }%
\providecommand \Eprint [0]{\href }%
\providecommand \doibase [0]{http://dx.doi.org/}%
\providecommand \selectlanguage [0]{\@gobble}%
\providecommand \bibinfo  [0]{\@secondoftwo}%
\providecommand \bibfield  [0]{\@secondoftwo}%
\providecommand \translation [1]{[#1]}%
\providecommand \BibitemOpen [0]{}%
\providecommand \bibitemStop [0]{}%
\providecommand \bibitemNoStop [0]{.\EOS\space}%
\providecommand \EOS [0]{\spacefactor3000\relax}%
\providecommand \BibitemShut  [1]{\csname bibitem#1\endcsname}%
\let\auto@bib@innerbib\@empty
\end{thebibliography}%


\begin{thebibliography}{43}%
\makeatletter
\providecommand \@ifxundefined [1]{%
 \@ifx{#1\undefined}
}%
\providecommand \@ifnum [1]{%
 \ifnum #1\expandafter \@firstoftwo
 \else \expandafter \@secondoftwo
 \fi
}%
\providecommand \@ifx [1]{%
 \ifx #1\expandafter \@firstoftwo
 \else \expandafter \@secondoftwo
 \fi
}%
\providecommand \natexlab [1]{#1}%
\providecommand \enquote  [1]{``#1''}%
\providecommand \bibnamefont  [1]{#1}%
\providecommand \bibfnamefont [1]{#1}%
\providecommand \citenamefont [1]{#1}%
\providecommand \href@noop [0]{\@secondoftwo}%
\providecommand \href [0]{\begingroup \@sanitize@url \@href}%
\providecommand \@href[1]{\@@startlink{#1}\@@href}%
\providecommand \@@href[1]{\endgroup#1\@@endlink}%
\providecommand \@sanitize@url [0]{\catcode `\\12\catcode `\$12\catcode
  `\&12\catcode `\#12\catcode `\^12\catcode `\_12\catcode `\%12\relax}%
\providecommand \@@startlink[1]{}%
\providecommand \@@endlink[0]{}%
\providecommand \url  [0]{\begingroup\@sanitize@url \@url }%
\providecommand \@url [1]{\endgroup\@href {#1}{\urlprefix }}%
\providecommand \urlprefix  [0]{URL }%
\providecommand \Eprint [0]{\href }%
\providecommand \doibase [0]{http://dx.doi.org/}%
\providecommand \selectlanguage [0]{\@gobble}%
\providecommand \bibinfo  [0]{\@secondoftwo}%
\providecommand \bibfield  [0]{\@secondoftwo}%
\providecommand \translation [1]{[#1]}%
\providecommand \BibitemOpen [0]{}%
\providecommand \bibitemStop [0]{}%
\providecommand \bibitemNoStop [0]{.\EOS\space}%
\providecommand \EOS [0]{\spacefactor3000\relax}%
\providecommand \BibitemShut  [1]{\csname bibitem#1\endcsname}%
\let\auto@bib@innerbib\@empty
\bibitem [{\citenamefont {Yang}\ and\ \citenamefont
  {Lee}(1952)}]{Yang-PR87.404--409-1952}%
  \BibitemOpen
  \bibfield  {author} {\bibinfo {author} {\bibfnamefont {C.~N.}\ \bibnamefont
  {Yang}}\ and\ \bibinfo {author} {\bibfnamefont {T.~D.}\ \bibnamefont {Lee}},\
  }\href {\doibase 10.1103/PhysRev.87.404} {\bibfield  {journal} {\bibinfo
  {journal} {Phys. Rev.}\ }\textbf {\bibinfo {volume} {87}},\ \bibinfo {pages}
  {404} (\bibinfo {year} {1952})}\BibitemShut {NoStop}%
\bibitem [{\citenamefont {Lee}\ and\ \citenamefont
  {Yang}(1952)}]{Lee-PR87.410--419-1952}%
  \BibitemOpen
  \bibfield  {author} {\bibinfo {author} {\bibfnamefont {T.~D.}\ \bibnamefont
  {Lee}}\ and\ \bibinfo {author} {\bibfnamefont {C.~N.}\ \bibnamefont {Yang}},\
  }\href {\doibase 10.1103/PhysRev.87.410} {\bibfield  {journal} {\bibinfo
  {journal} {Phys. Rev.}\ }\textbf {\bibinfo {volume} {87}},\ \bibinfo {pages}
  {410} (\bibinfo {year} {1952})}\BibitemShut {NoStop}%
\bibitem [{\citenamefont {Fisher}(1978)}]{Fisher-PRL40.1610--1613-1978}%
  \BibitemOpen
  \bibfield  {author} {\bibinfo {author} {\bibfnamefont {M.~E.}\ \bibnamefont
  {Fisher}},\ }\href {\doibase 10.1103/PhysRevLett.40.1610} {\bibfield
  {journal} {\bibinfo  {journal} {Phys. Rev. Lett.}\ }\textbf {\bibinfo
  {volume} {40}},\ \bibinfo {pages} {1610} (\bibinfo {year}
  {1978})}\BibitemShut {NoStop}%
\bibitem [{\citenamefont {Kortman}\ and\ \citenamefont
  {Griffiths}(1971)}]{Kortman-PRL27.1439--1442-1971}%
  \BibitemOpen
  \bibfield  {author} {\bibinfo {author} {\bibfnamefont {P.~J.}\ \bibnamefont
  {Kortman}}\ and\ \bibinfo {author} {\bibfnamefont {R.~B.}\ \bibnamefont
  {Griffiths}},\ }\href {\doibase 10.1103/PhysRevLett.27.1439} {\bibfield
  {journal} {\bibinfo  {journal} {Phys. Rev. Lett.}\ }\textbf {\bibinfo
  {volume} {27}},\ \bibinfo {pages} {1439} (\bibinfo {year}
  {1971})}\BibitemShut {NoStop}%
\bibitem [{\citenamefont {Wei}\ and\ \citenamefont
  {Liu}(2012)}]{Wei-PRL109.185701-2012}%
  \BibitemOpen
  \bibfield  {author} {\bibinfo {author} {\bibfnamefont {B.-B.}\ \bibnamefont
  {Wei}}\ and\ \bibinfo {author} {\bibfnamefont {R.-B.}\ \bibnamefont {Liu}},\
  }\href {\doibase 10.1103/PhysRevLett.109.185701} {\bibfield  {journal}
  {\bibinfo  {journal} {Phys. Rev. Lett.}\ }\textbf {\bibinfo {volume} {109}},\
  \bibinfo {pages} {185701} (\bibinfo {year} {2012})}\BibitemShut {NoStop}%
\bibitem [{\citenamefont
  {Schlosshauer}(2005)}]{Schlosshauer-RMP76.1267--1305-2005}%
  \BibitemOpen
  \bibfield  {author} {\bibinfo {author} {\bibfnamefont {M.}~\bibnamefont
  {Schlosshauer}},\ }\href {\doibase 10.1103/RevModPhys.76.1267} {\bibfield
  {journal} {\bibinfo  {journal} {Rev. Mod. Phys.}\ }\textbf {\bibinfo {volume}
  {76}},\ \bibinfo {pages} {1267} (\bibinfo {year} {2005})}\BibitemShut
  {NoStop}%
\bibitem [{\citenamefont {Liu}\ \emph {et~al.}(2007)\citenamefont {Liu},
  \citenamefont {Yao},\ and\ \citenamefont {Sham}}]{Liu-NJP9.226-2007}%
  \BibitemOpen
  \bibfield  {author} {\bibinfo {author} {\bibfnamefont {R.-B.}\ \bibnamefont
  {Liu}}, \bibinfo {author} {\bibfnamefont {W.}~\bibnamefont {Yao}}, \ and\
  \bibinfo {author} {\bibfnamefont {L.~J.}\ \bibnamefont {Sham}},\ }\href
  {\doibase 10.1088/1367-2630/16/2/023019} {\bibfield  {journal} {\bibinfo
  {journal} {New J. Phys.}\ }\textbf {\bibinfo {volume} {9}},\ \bibinfo {pages}
  {226} (\bibinfo {year} {2007})}\BibitemShut {NoStop}%
\bibitem [{\citenamefont {Peng}\ \emph {et~al.}(2015)\citenamefont {Peng},
  \citenamefont {Zhou}, \citenamefont {Wei}, \citenamefont {Cui}, \citenamefont
  {Du},\ and\ \citenamefont {Liu}}]{Peng2015}%
  \BibitemOpen
  \bibfield  {author} {\bibinfo {author} {\bibfnamefont {X.}~\bibnamefont
  {Peng}}, \bibinfo {author} {\bibfnamefont {H.}~\bibnamefont {Zhou}}, \bibinfo
  {author} {\bibfnamefont {B.-B.}\ \bibnamefont {Wei}}, \bibinfo {author}
  {\bibfnamefont {J.}~\bibnamefont {Cui}}, \bibinfo {author} {\bibfnamefont
  {J.}~\bibnamefont {Du}}, \ and\ \bibinfo {author} {\bibfnamefont {R.-B.}\
  \bibnamefont {Liu}},\ }\href {\doibase 10.1103/PhysRevLett.114.010601}
  {\bibfield  {journal} {\bibinfo  {journal} {Phys. Rev. Lett.}\ }\textbf
  {\bibinfo {volume} {114}},\ \bibinfo {pages} {010601} (\bibinfo {year}
  {2015})}\BibitemShut {NoStop}%
\bibitem [{\citenamefont {Asano}(1968)}]{Asano-PTP40.1328-1336-1968}%
  \BibitemOpen
  \bibfield  {author} {\bibinfo {author} {\bibfnamefont {T.}~\bibnamefont
  {Asano}},\ }\href {\doibase 10.1143/PTP.40.1328} {\bibfield  {journal}
  {\bibinfo  {journal} {Prog. Theor. Phys.}\ }\textbf {\bibinfo {volume}
  {40}},\ \bibinfo {pages} {1328} (\bibinfo {year} {1968})}\BibitemShut
  {NoStop}%
\bibitem [{\citenamefont
  {Suzuki}(1968{\natexlab{a}})}]{Suzuki-PTP40.1246-1256-1968}%
  \BibitemOpen
  \bibfield  {author} {\bibinfo {author} {\bibfnamefont {M.}~\bibnamefont
  {Suzuki}},\ }\href {\doibase 10.1143/PTP.40.1246} {\bibfield  {journal}
  {\bibinfo  {journal} {Prog. Theor. Phys.}\ }\textbf {\bibinfo {volume}
  {40}},\ \bibinfo {pages} {1246} (\bibinfo {year}
  {1968}{\natexlab{a}})}\BibitemShut {NoStop}%
\bibitem [{\citenamefont {Griffiths}(1969)}]{Griffiths-JMP10.1559-1565-1969}%
  \BibitemOpen
  \bibfield  {author} {\bibinfo {author} {\bibfnamefont {R.~B.}\ \bibnamefont
  {Griffiths}},\ }\href {\doibase 10.1063/1.1665005} {\bibfield  {journal}
  {\bibinfo  {journal} {J. Math. Phys.}\ }\textbf {\bibinfo {volume} {10}},\
  \bibinfo {pages} {1559} (\bibinfo {year} {1969})}\BibitemShut {NoStop}%
\bibitem [{\citenamefont {Kim}(2004)}]{Kim2004}%
  \BibitemOpen
  \bibfield  {author} {\bibinfo {author} {\bibfnamefont {S.-Y.}\ \bibnamefont
  {Kim}},\ }\href {\doibase 10.1103/PhysRevLett.93.130604} {\bibfield
  {journal} {\bibinfo  {journal} {Phys. Rev. Lett.}\ }\textbf {\bibinfo
  {volume} {93}},\ \bibinfo {pages} {130604} (\bibinfo {year}
  {2004})}\BibitemShut {NoStop}%
\bibitem [{\citenamefont {Brandner}\ \emph {et~al.}(2017)\citenamefont
  {Brandner}, \citenamefont {Maisi}, \citenamefont {Pekola}, \citenamefont
  {Garrahan},\ and\ \citenamefont {Flindt}}]{Brandner2017}%
  \BibitemOpen
  \bibfield  {author} {\bibinfo {author} {\bibfnamefont {K.}~\bibnamefont
  {Brandner}}, \bibinfo {author} {\bibfnamefont {V.~F.}\ \bibnamefont {Maisi}},
  \bibinfo {author} {\bibfnamefont {J.~P.}\ \bibnamefont {Pekola}}, \bibinfo
  {author} {\bibfnamefont {J.~P.}\ \bibnamefont {Garrahan}}, \ and\ \bibinfo
  {author} {\bibfnamefont {C.}~\bibnamefont {Flindt}},\ }\href {\doibase
  10.1103/PhysRevLett.118.180601} {\bibfield  {journal} {\bibinfo  {journal}
  {Phys. Rev. Lett.}\ }\textbf {\bibinfo {volume} {118}},\ \bibinfo {pages}
  {180601} (\bibinfo {year} {2017})}\BibitemShut {NoStop}%
\bibitem [{\citenamefont
  {Suzuki}(1968{\natexlab{b}})}]{Suzuki-JMP9.2064-2068-1968}%
  \BibitemOpen
  \bibfield  {author} {\bibinfo {author} {\bibfnamefont {M.}~\bibnamefont
  {Suzuki}},\ }\href {\doibase 10.1063/1.1664546} {\bibfield  {journal}
  {\bibinfo  {journal} {J. Math. Phys.}\ }\textbf {\bibinfo {volume} {9}},\
  \bibinfo {pages} {2064} (\bibinfo {year} {1968}{\natexlab{b}})}\BibitemShut
  {NoStop}%
\bibitem [{\citenamefont {Suzuki}\ and\ \citenamefont
  {Fisher}(1971)}]{Suzuki-JMP12.235-246-1971}%
  \BibitemOpen
  \bibfield  {author} {\bibinfo {author} {\bibfnamefont {M.}~\bibnamefont
  {Suzuki}}\ and\ \bibinfo {author} {\bibfnamefont {M.~E.}\ \bibnamefont
  {Fisher}},\ }\href {\doibase 10.1063/1.1665583} {\bibfield  {journal}
  {\bibinfo  {journal} {J. Math. Phys.}\ }\textbf {\bibinfo {volume} {12}},\
  \bibinfo {pages} {235} (\bibinfo {year} {1971})}\BibitemShut {NoStop}%
\bibitem [{\citenamefont {Kurtze}\ and\ \citenamefont
  {Fisher}(1978)}]{Kurtze-JSP19.205--218-1978}%
  \BibitemOpen
  \bibfield  {author} {\bibinfo {author} {\bibfnamefont {D.~A.}\ \bibnamefont
  {Kurtze}}\ and\ \bibinfo {author} {\bibfnamefont {M.~E.}\ \bibnamefont
  {Fisher}},\ }\href {\doibase 10.1007/BF01011723} {\bibfield  {journal}
  {\bibinfo  {journal} {J. Stat. Phys.}\ }\textbf {\bibinfo {volume} {19}},\
  \bibinfo {pages} {205} (\bibinfo {year} {1978})}\BibitemShut {NoStop}%
\bibitem [{\citenamefont {Kitagawa}\ and\ \citenamefont
  {Ueda}(1993)}]{Kitagawa1993}%
  \BibitemOpen
  \bibfield  {author} {\bibinfo {author} {\bibfnamefont {M.}~\bibnamefont
  {Kitagawa}}\ and\ \bibinfo {author} {\bibfnamefont {M.}~\bibnamefont
  {Ueda}},\ }\href {\doibase 10.1103/PhysRevA.47.5138} {\bibfield  {journal}
  {\bibinfo  {journal} {Phys. Rev. A}\ }\textbf {\bibinfo {volume} {47}},\
  \bibinfo {pages} {5138} (\bibinfo {year} {1993})}\BibitemShut {NoStop}%
\bibitem [{\citenamefont {Wineland}\ \emph {et~al.}(1994)\citenamefont
  {Wineland}, \citenamefont {Bollinger}, \citenamefont {Itano},\ and\
  \citenamefont {Heinzen}}]{Wineland1994}%
  \BibitemOpen
  \bibfield  {author} {\bibinfo {author} {\bibfnamefont {D.~J.}\ \bibnamefont
  {Wineland}}, \bibinfo {author} {\bibfnamefont {J.~J.}\ \bibnamefont
  {Bollinger}}, \bibinfo {author} {\bibfnamefont {W.~M.}\ \bibnamefont
  {Itano}}, \ and\ \bibinfo {author} {\bibfnamefont {D.~J.}\ \bibnamefont
  {Heinzen}},\ }\href {\doibase 10.1103/PhysRevA.50.67} {\bibfield  {journal}
  {\bibinfo  {journal} {Phys. Rev. A}\ }\textbf {\bibinfo {volume} {50}},\
  \bibinfo {pages} {67} (\bibinfo {year} {1994})}\BibitemShut {NoStop}%
\bibitem [{\citenamefont {Sorensen}\ \emph {et~al.}(2001)\citenamefont
  {Sorensen}, \citenamefont {Duan}, \citenamefont {Cirac},\ and\ \citenamefont
  {Zoller}}]{Sorensen2001}%
  \BibitemOpen
  \bibfield  {author} {\bibinfo {author} {\bibfnamefont {A.}~\bibnamefont
  {Sorensen}}, \bibinfo {author} {\bibfnamefont {L.~M.}\ \bibnamefont {Duan}},
  \bibinfo {author} {\bibfnamefont {J.~I.}\ \bibnamefont {Cirac}}, \ and\
  \bibinfo {author} {\bibfnamefont {P.}~\bibnamefont {Zoller}},\ }\href
  {\doibase 10.1038/35051038} {\bibfield  {journal} {\bibinfo  {journal}
  {Nature}\ }\textbf {\bibinfo {volume} {409}},\ \bibinfo {pages} {63}
  (\bibinfo {year} {2001})}\BibitemShut {NoStop}%
\bibitem [{\citenamefont {T\'oth}\ \emph {et~al.}(2007)\citenamefont {T\'oth},
  \citenamefont {Knapp}, \citenamefont {G\"uhne},\ and\ \citenamefont
  {Briegel}}]{Toth2007}%
  \BibitemOpen
  \bibfield  {author} {\bibinfo {author} {\bibfnamefont {G.}~\bibnamefont
  {T\'oth}}, \bibinfo {author} {\bibfnamefont {C.}~\bibnamefont {Knapp}},
  \bibinfo {author} {\bibfnamefont {O.}~\bibnamefont {G\"uhne}}, \ and\
  \bibinfo {author} {\bibfnamefont {H.~J.}\ \bibnamefont {Briegel}},\ }\href
  {\doibase 10.1103/PhysRevLett.99.250405} {\bibfield  {journal} {\bibinfo
  {journal} {Phys. Rev. Lett.}\ }\textbf {\bibinfo {volume} {99}},\ \bibinfo
  {pages} {250405} (\bibinfo {year} {2007})}\BibitemShut {NoStop}%
\bibitem [{\citenamefont {T\'oth}\ \emph {et~al.}(2009)\citenamefont {T\'oth},
  \citenamefont {Knapp}, \citenamefont {G\"uhne},\ and\ \citenamefont
  {Briegel}}]{Toth2009}%
  \BibitemOpen
  \bibfield  {author} {\bibinfo {author} {\bibfnamefont {G.}~\bibnamefont
  {T\'oth}}, \bibinfo {author} {\bibfnamefont {C.}~\bibnamefont {Knapp}},
  \bibinfo {author} {\bibfnamefont {O.}~\bibnamefont {G\"uhne}}, \ and\
  \bibinfo {author} {\bibfnamefont {H.~J.}\ \bibnamefont {Briegel}},\ }\href
  {\doibase 10.1103/PhysRevA.79.042334} {\bibfield  {journal} {\bibinfo
  {journal} {Phys. Rev. A}\ }\textbf {\bibinfo {volume} {79}},\ \bibinfo
  {pages} {042334} (\bibinfo {year} {2009})}\BibitemShut {NoStop}%
\bibitem [{\citenamefont {S\o{}rensen}\ and\ \citenamefont
  {M\o{}lmer}(1999)}]{Sorensen1999}%
  \BibitemOpen
  \bibfield  {author} {\bibinfo {author} {\bibfnamefont {A.}~\bibnamefont
  {S\o{}rensen}}\ and\ \bibinfo {author} {\bibfnamefont {K.}~\bibnamefont
  {M\o{}lmer}},\ }\href {\doibase 10.1103/PhysRevLett.83.2274} {\bibfield
  {journal} {\bibinfo  {journal} {Phys. Rev. Lett.}\ }\textbf {\bibinfo
  {volume} {83}},\ \bibinfo {pages} {2274} (\bibinfo {year}
  {1999})}\BibitemShut {NoStop}%
\bibitem [{\citenamefont {Wang}\ and\ \citenamefont
  {Sanders}(2003)}]{Wang2003}%
  \BibitemOpen
  \bibfield  {author} {\bibinfo {author} {\bibfnamefont {X.}~\bibnamefont
  {Wang}}\ and\ \bibinfo {author} {\bibfnamefont {B.~C.}\ \bibnamefont
  {Sanders}},\ }\href {\doibase 10.1103/PhysRevA.68.012101} {\bibfield
  {journal} {\bibinfo  {journal} {Phys. Rev. A}\ }\textbf {\bibinfo {volume}
  {68}},\ \bibinfo {pages} {012101} (\bibinfo {year} {2003})}\BibitemShut
  {NoStop}%
\bibitem [{\citenamefont {Ma}\ \emph {et~al.}(2011)\citenamefont {Ma},
  \citenamefont {Wang}, \citenamefont {Sun},\ and\ \citenamefont
  {Nori}}]{Ma2011}%
  \BibitemOpen
  \bibfield  {author} {\bibinfo {author} {\bibfnamefont {J.}~\bibnamefont
  {Ma}}, \bibinfo {author} {\bibfnamefont {X.}~\bibnamefont {Wang}}, \bibinfo
  {author} {\bibfnamefont {C.}~\bibnamefont {Sun}}, \ and\ \bibinfo {author}
  {\bibfnamefont {F.}~\bibnamefont {Nori}},\ }\href {\doibase
  https://doi.org/10.1016/j.physrep.2011.08.003} {\bibfield  {journal}
  {\bibinfo  {journal} {Phys. Rep.}\ }\textbf {\bibinfo {volume} {509}},\
  \bibinfo {pages} {89 } (\bibinfo {year} {2011})}\BibitemShut {NoStop}%
\bibitem [{\citenamefont {Ulam-Orgikh}\ and\ \citenamefont
  {Kitagawa}(2001)}]{Ulam-Orgikh2001}%
  \BibitemOpen
  \bibfield  {author} {\bibinfo {author} {\bibfnamefont {D.}~\bibnamefont
  {Ulam-Orgikh}}\ and\ \bibinfo {author} {\bibfnamefont {M.}~\bibnamefont
  {Kitagawa}},\ }\href {\doibase 10.1103/PhysRevA.64.052106} {\bibfield
  {journal} {\bibinfo  {journal} {Phys. Rev. A}\ }\textbf {\bibinfo {volume}
  {64}},\ \bibinfo {pages} {052106} (\bibinfo {year} {2001})}\BibitemShut
  {NoStop}%
\bibitem [{\citenamefont {Jin}\ and\ \citenamefont {Kim}(2007)}]{Jin2007}%
  \BibitemOpen
  \bibfield  {author} {\bibinfo {author} {\bibfnamefont {G.-R.}\ \bibnamefont
  {Jin}}\ and\ \bibinfo {author} {\bibfnamefont {S.~W.}\ \bibnamefont {Kim}},\
  }\href {\doibase 10.1103/PhysRevA.76.043621} {\bibfield  {journal} {\bibinfo
  {journal} {Phys. Rev. A}\ }\textbf {\bibinfo {volume} {76}},\ \bibinfo
  {pages} {043621} (\bibinfo {year} {2007})}\BibitemShut {NoStop}%
\bibitem [{\citenamefont {Jafarpour}\ and\ \citenamefont
  {Akhound}(2008)}]{Jafarpour}%
  \BibitemOpen
  \bibfield  {author} {\bibinfo {author} {\bibfnamefont {M.}~\bibnamefont
  {Jafarpour}}\ and\ \bibinfo {author} {\bibfnamefont {A.}~\bibnamefont
  {Akhound}},\ }\href {\doibase 10.1016/j.physleta.2007.12.021} {\bibfield
  {journal} {\bibinfo  {journal} {Phys. Lett. A}\ }\textbf {\bibinfo {volume}
  {372}},\ \bibinfo {pages} {2374} (\bibinfo {year} {2008})}\BibitemShut
  {NoStop}%
\bibitem [{\citenamefont {Messikh}\ \emph {et~al.}(2003)\citenamefont
  {Messikh}, \citenamefont {Ficek},\ and\ \citenamefont
  {Wahiddin}}]{Messikh2003}%
  \BibitemOpen
  \bibfield  {author} {\bibinfo {author} {\bibfnamefont {A.}~\bibnamefont
  {Messikh}}, \bibinfo {author} {\bibfnamefont {Z.}~\bibnamefont {Ficek}}, \
  and\ \bibinfo {author} {\bibfnamefont {M.~R.~B.}\ \bibnamefont {Wahiddin}},\
  }\href {\doibase 10.1103/PhysRevA.68.064301} {\bibfield  {journal} {\bibinfo
  {journal} {Phys. Rev. A}\ }\textbf {\bibinfo {volume} {68}},\ \bibinfo
  {pages} {064301} (\bibinfo {year} {2003})}\BibitemShut {NoStop}%
\bibitem [{\citenamefont {Wang}(2004)}]{Wang2004}%
  \BibitemOpen
  \bibfield  {author} {\bibinfo {author} {\bibfnamefont {X.}~\bibnamefont
  {Wang}},\ }\href {\doibase https://doi.org/10.1016/j.physleta.2004.08.019}
  {\bibfield  {journal} {\bibinfo  {journal} {Phys. Lett. A}\ }\textbf
  {\bibinfo {volume} {331}},\ \bibinfo {pages} {164 } (\bibinfo {year}
  {2004})}\BibitemShut {NoStop}%
\bibitem [{\citenamefont {Yin}\ \emph {et~al.}(2010)\citenamefont {Yin},
  \citenamefont {Wang}, \citenamefont {Ma},\ and\ \citenamefont
  {Wang}}]{Yin2010}%
  \BibitemOpen
  \bibfield  {author} {\bibinfo {author} {\bibfnamefont {X.}~\bibnamefont
  {Yin}}, \bibinfo {author} {\bibfnamefont {X.}~\bibnamefont {Wang}}, \bibinfo
  {author} {\bibfnamefont {J.}~\bibnamefont {Ma}}, \ and\ \bibinfo {author}
  {\bibfnamefont {X.}~\bibnamefont {Wang}},\ }\href {\doibase
  10.1088/0953-4075/44/1/015501} {\bibfield  {journal} {\bibinfo  {journal} {J.
  Phys. B: At., Mol. Opt. Phys.}\ }\textbf {\bibinfo {volume} {44}},\ \bibinfo
  {pages} {015501} (\bibinfo {year} {2010})}\BibitemShut {NoStop}%
\bibitem [{\citenamefont {Genes}\ \emph {et~al.}(2003)\citenamefont {Genes},
  \citenamefont {Berman},\ and\ \citenamefont {Rojo}}]{Genes2003}%
  \BibitemOpen
  \bibfield  {author} {\bibinfo {author} {\bibfnamefont {C.}~\bibnamefont
  {Genes}}, \bibinfo {author} {\bibfnamefont {P.~R.}\ \bibnamefont {Berman}}, \
  and\ \bibinfo {author} {\bibfnamefont {A.~G.}\ \bibnamefont {Rojo}},\ }\href
  {\doibase 10.1103/PhysRevA.68.043809} {\bibfield  {journal} {\bibinfo
  {journal} {Phys. Rev. A}\ }\textbf {\bibinfo {volume} {68}},\ \bibinfo
  {pages} {043809} (\bibinfo {year} {2003})}\BibitemShut {NoStop}%
\bibitem [{\citenamefont {Fernholz}\ \emph {et~al.}(2008)\citenamefont
  {Fernholz}, \citenamefont {Krauter}, \citenamefont {Jensen}, \citenamefont
  {Sherson}, \citenamefont {S\o{}rensen},\ and\ \citenamefont
  {Polzik}}]{Fernholz2008}%
  \BibitemOpen
  \bibfield  {author} {\bibinfo {author} {\bibfnamefont {T.}~\bibnamefont
  {Fernholz}}, \bibinfo {author} {\bibfnamefont {H.}~\bibnamefont {Krauter}},
  \bibinfo {author} {\bibfnamefont {K.}~\bibnamefont {Jensen}}, \bibinfo
  {author} {\bibfnamefont {J.~F.}\ \bibnamefont {Sherson}}, \bibinfo {author}
  {\bibfnamefont {A.~S.}\ \bibnamefont {S\o{}rensen}}, \ and\ \bibinfo {author}
  {\bibfnamefont {E.~S.}\ \bibnamefont {Polzik}},\ }\href {\doibase
  10.1103/PhysRevLett.101.073601} {\bibfield  {journal} {\bibinfo  {journal}
  {Phys. Rev. Lett.}\ }\textbf {\bibinfo {volume} {101}},\ \bibinfo {pages}
  {073601} (\bibinfo {year} {2008})}\BibitemShut {NoStop}%
\bibitem [{\citenamefont {Takano}\ \emph {et~al.}(2009)\citenamefont {Takano},
  \citenamefont {Fuyama}, \citenamefont {Namiki},\ and\ \citenamefont
  {Takahashi}}]{Takano2009}%
  \BibitemOpen
  \bibfield  {author} {\bibinfo {author} {\bibfnamefont {T.}~\bibnamefont
  {Takano}}, \bibinfo {author} {\bibfnamefont {M.}~\bibnamefont {Fuyama}},
  \bibinfo {author} {\bibfnamefont {R.}~\bibnamefont {Namiki}}, \ and\ \bibinfo
  {author} {\bibfnamefont {Y.}~\bibnamefont {Takahashi}},\ }\href {\doibase
  10.1103/PhysRevLett.102.033601} {\bibfield  {journal} {\bibinfo  {journal}
  {Phys. Rev. Lett.}\ }\textbf {\bibinfo {volume} {102}},\ \bibinfo {pages}
  {033601} (\bibinfo {year} {2009})}\BibitemShut {NoStop}%
\bibitem [{\citenamefont {Wineland}\ \emph {et~al.}(1992)\citenamefont
  {Wineland}, \citenamefont {Bollinger}, \citenamefont {Itano}, \citenamefont
  {Moore},\ and\ \citenamefont {Heinzen}}]{Wineland1992}%
  \BibitemOpen
  \bibfield  {author} {\bibinfo {author} {\bibfnamefont {D.~J.}\ \bibnamefont
  {Wineland}}, \bibinfo {author} {\bibfnamefont {J.~J.}\ \bibnamefont
  {Bollinger}}, \bibinfo {author} {\bibfnamefont {W.~M.}\ \bibnamefont
  {Itano}}, \bibinfo {author} {\bibfnamefont {F.~L.}\ \bibnamefont {Moore}}, \
  and\ \bibinfo {author} {\bibfnamefont {D.~J.}\ \bibnamefont {Heinzen}},\
  }\href {\doibase 10.1103/PhysRevA.46.R6797} {\bibfield  {journal} {\bibinfo
  {journal} {Phys. Rev. A}\ }\textbf {\bibinfo {volume} {46}},\ \bibinfo
  {pages} {R6797} (\bibinfo {year} {1992})}\BibitemShut {NoStop}%
\bibitem [{\citenamefont {Cronin}\ \emph {et~al.}(2009)\citenamefont {Cronin},
  \citenamefont {Schmiedmayer},\ and\ \citenamefont {Pritchard}}]{Cronin2009}%
  \BibitemOpen
  \bibfield  {author} {\bibinfo {author} {\bibfnamefont {A.~D.}\ \bibnamefont
  {Cronin}}, \bibinfo {author} {\bibfnamefont {J.}~\bibnamefont
  {Schmiedmayer}}, \ and\ \bibinfo {author} {\bibfnamefont {D.~E.}\
  \bibnamefont {Pritchard}},\ }\href {\doibase 10.1103/RevModPhys.81.1051}
  {\bibfield  {journal} {\bibinfo  {journal} {Rev. Mod. Phys.}\ }\textbf
  {\bibinfo {volume} {81}},\ \bibinfo {pages} {1051} (\bibinfo {year}
  {2009})}\BibitemShut {NoStop}%
\bibitem [{\citenamefont {Bollinger}\ \emph {et~al.}(1996)\citenamefont
  {Bollinger}, \citenamefont {Itano}, \citenamefont {Wineland},\ and\
  \citenamefont {Heinzen}}]{Bollinger1996}%
  \BibitemOpen
  \bibfield  {author} {\bibinfo {author} {\bibfnamefont {J.~J.}\ \bibnamefont
  {Bollinger}}, \bibinfo {author} {\bibfnamefont {W.~M.}\ \bibnamefont
  {Itano}}, \bibinfo {author} {\bibfnamefont {D.~J.}\ \bibnamefont {Wineland}},
  \ and\ \bibinfo {author} {\bibfnamefont {D.~J.}\ \bibnamefont {Heinzen}},\
  }\href {\doibase 10.1103/PhysRevA.54.R4649} {\bibfield  {journal} {\bibinfo
  {journal} {Phys. Rev. A}\ }\textbf {\bibinfo {volume} {54}},\ \bibinfo
  {pages} {R4649} (\bibinfo {year} {1996})}\BibitemShut {NoStop}%
\bibitem [{\citenamefont {D\"oring}\ \emph {et~al.}(2010)\citenamefont
  {D\"oring}, \citenamefont {McDonald}, \citenamefont {Debs}, \citenamefont
  {Figl}, \citenamefont {Altin}, \citenamefont {Bachor}, \citenamefont
  {Robins},\ and\ \citenamefont {Close}}]{Doering2010}%
  \BibitemOpen
  \bibfield  {author} {\bibinfo {author} {\bibfnamefont {D.}~\bibnamefont
  {D\"oring}}, \bibinfo {author} {\bibfnamefont {G.}~\bibnamefont {McDonald}},
  \bibinfo {author} {\bibfnamefont {J.~E.}\ \bibnamefont {Debs}}, \bibinfo
  {author} {\bibfnamefont {C.}~\bibnamefont {Figl}}, \bibinfo {author}
  {\bibfnamefont {P.~A.}\ \bibnamefont {Altin}}, \bibinfo {author}
  {\bibfnamefont {H.-A.}\ \bibnamefont {Bachor}}, \bibinfo {author}
  {\bibfnamefont {N.~P.}\ \bibnamefont {Robins}}, \ and\ \bibinfo {author}
  {\bibfnamefont {J.~D.}\ \bibnamefont {Close}},\ }\href {\doibase
  10.1103/PhysRevA.81.043633} {\bibfield  {journal} {\bibinfo  {journal} {Phys.
  Rev. A}\ }\textbf {\bibinfo {volume} {81}},\ \bibinfo {pages} {043633}
  (\bibinfo {year} {2010})}\BibitemShut {NoStop}%
\bibitem [{\citenamefont {Andr\'e}\ \emph {et~al.}(2004)\citenamefont
  {Andr\'e}, \citenamefont {S\o{}rensen},\ and\ \citenamefont
  {Lukin}}]{Andre2004}%
  \BibitemOpen
  \bibfield  {author} {\bibinfo {author} {\bibfnamefont {A.}~\bibnamefont
  {Andr\'e}}, \bibinfo {author} {\bibfnamefont {A.~S.}\ \bibnamefont
  {S\o{}rensen}}, \ and\ \bibinfo {author} {\bibfnamefont {M.~D.}\ \bibnamefont
  {Lukin}},\ }\href {\doibase 10.1103/PhysRevLett.92.230801} {\bibfield
  {journal} {\bibinfo  {journal} {Phys. Rev. Lett.}\ }\textbf {\bibinfo
  {volume} {92}},\ \bibinfo {pages} {230801} (\bibinfo {year}
  {2004})}\BibitemShut {NoStop}%
\bibitem [{\citenamefont {Meiser}\ \emph {et~al.}(2008)\citenamefont {Meiser},
  \citenamefont {Ye},\ and\ \citenamefont {Holland}}]{Meiser2008}%
  \BibitemOpen
  \bibfield  {author} {\bibinfo {author} {\bibfnamefont {D.}~\bibnamefont
  {Meiser}}, \bibinfo {author} {\bibfnamefont {J.}~\bibnamefont {Ye}}, \ and\
  \bibinfo {author} {\bibfnamefont {M.~J.}\ \bibnamefont {Holland}},\ }\href
  {\doibase 10.1088/1367-2630/10/7/073014} {\bibfield  {journal} {\bibinfo
  {journal} {New J. Phys.}\ }\textbf {\bibinfo {volume} {10}},\ \bibinfo
  {pages} {073014} (\bibinfo {year} {2008})}\BibitemShut {NoStop}%
\bibitem [{\citenamefont {Walls}\ and\ \citenamefont {Zoller}(1981)}]{Walls}%
  \BibitemOpen
  \bibfield  {author} {\bibinfo {author} {\bibfnamefont {D.~F.}\ \bibnamefont
  {Walls}}\ and\ \bibinfo {author} {\bibfnamefont {P.}~\bibnamefont {Zoller}},\
  }\href {\doibase 10.1016/0375-9601(81)90238-3} {\bibfield  {journal}
  {\bibinfo  {journal} {Phys. Lett. A}\ }\textbf {\bibinfo {volume} {85}},\
  \bibinfo {pages} {118} (\bibinfo {year} {1981})}\BibitemShut {NoStop}%
\bibitem [{\citenamefont {Goda}\ \emph {et~al.}(2008)\citenamefont {Goda},
  \citenamefont {Miyakawa}, \citenamefont {Mikhailov}, \citenamefont {Saraf},
  \citenamefont {Adhikari}, \citenamefont {McKenzie}, \citenamefont {Ward},
  \citenamefont {Vass}, \citenamefont {Weinstein},\ and\ \citenamefont
  {Mavalvala}}]{Goda}%
  \BibitemOpen
  \bibfield  {author} {\bibinfo {author} {\bibfnamefont {K.}~\bibnamefont
  {Goda}}, \bibinfo {author} {\bibfnamefont {O.}~\bibnamefont {Miyakawa}},
  \bibinfo {author} {\bibfnamefont {E.~E.}\ \bibnamefont {Mikhailov}}, \bibinfo
  {author} {\bibfnamefont {S.}~\bibnamefont {Saraf}}, \bibinfo {author}
  {\bibfnamefont {R.}~\bibnamefont {Adhikari}}, \bibinfo {author}
  {\bibfnamefont {K.}~\bibnamefont {McKenzie}}, \bibinfo {author}
  {\bibfnamefont {R.}~\bibnamefont {Ward}}, \bibinfo {author} {\bibfnamefont
  {S.}~\bibnamefont {Vass}}, \bibinfo {author} {\bibfnamefont {A.~J.}\
  \bibnamefont {Weinstein}}, \ and\ \bibinfo {author} {\bibfnamefont
  {N.}~\bibnamefont {Mavalvala}},\ }\href {\doibase 10.1038/nphys920}
  {\bibfield  {journal} {\bibinfo  {journal} {Nature Physics}\ }\textbf {\bibinfo
  {volume} {4}},\ \bibinfo {pages} {472} (\bibinfo {year} {2008})}\BibitemShut
  {NoStop}%
\bibitem [{\citenamefont {Yin}\ \emph {et~al.}(2012)\citenamefont {Yin},
  \citenamefont {Ma}, \citenamefont {Wang},\ and\ \citenamefont
  {Nori}}]{Yin2012}%
  \BibitemOpen
  \bibfield  {author} {\bibinfo {author} {\bibfnamefont {X.}~\bibnamefont
  {Yin}}, \bibinfo {author} {\bibfnamefont {J.}~\bibnamefont {Ma}}, \bibinfo
  {author} {\bibfnamefont {X.}~\bibnamefont {Wang}}, \ and\ \bibinfo {author}
  {\bibfnamefont {F.}~\bibnamefont {Nori}},\ }\href {\doibase
  10.1103/PhysRevA.86.012308} {\bibfield  {journal} {\bibinfo  {journal} {Phys.
  Rev. A}\ }\textbf {\bibinfo {volume} {86}},\ \bibinfo {pages} {012308}
  (\bibinfo {year} {2012})}\BibitemShut {NoStop}%
\bibitem [{\citenamefont {Bennett}\ \emph {et~al.}(1996)\citenamefont
  {Bennett}, \citenamefont {DiVincenzo}, \citenamefont {Smolin},\ and\
  \citenamefont {Wootters}}]{Bennett-PRA54.3824--3851-1996}%
  \BibitemOpen
  \bibfield  {author} {\bibinfo {author} {\bibfnamefont {C.~H.}\ \bibnamefont
  {Bennett}}, \bibinfo {author} {\bibfnamefont {D.~P.}\ \bibnamefont
  {DiVincenzo}}, \bibinfo {author} {\bibfnamefont {J.~A.}\ \bibnamefont
  {Smolin}}, \ and\ \bibinfo {author} {\bibfnamefont {W.~K.}\ \bibnamefont
  {Wootters}},\ }\href {\doibase 10.1103/PhysRevA.54.3824} {\bibfield
  {journal} {\bibinfo  {journal} {Phys. Rev. A}\ }\textbf {\bibinfo {volume}
  {54}},\ \bibinfo {pages} {3824} (\bibinfo {year} {1996})}\BibitemShut
  {NoStop}%
\end{thebibliography}
%


\end{document}